\documentclass[11pt,letterpaper]{article}
\usepackage{fullpage}
\usepackage{url}
\usepackage{algorithm}
\usepackage{amsmath}
\usepackage{amssymb}
\usepackage{amsthm}
\usepackage{hyperref}
\usepackage{cmap}
\usepackage{cite}
\usepackage{etex,etoolbox}
\usepackage{amsthm}

\makeatletter
\providecommand{\@fourthoffour}[4]{#4}
\def\fixstatement#1{%
  \AtEndEnvironment{#1}{%
    \xdef\pat@label{\expandafter\expandafter\expandafter
      \@fourthoffour\csname#1\endcsname\space\@currentlabel}}}

\globtoksblk\prooftoks{1000}
\newcounter{proofcount}

\long\def\proofatend#1\endproofatend{%
  \begin{proof}
    See Appendix \ref{proofsappendix}.
  \end{proof}
  \edef\next{\noexpand\begin{proof}[Proof of \pat@label]}%
  \toks\numexpr\prooftoks+\value{proofcount}\relax=\expandafter{\next#1\end{proof}}
  \stepcounter{proofcount}}

\def\printproofs{%
  \count@=\z@
  \loop
    \the\toks\numexpr\prooftoks+\count@\relax
    \ifnum\count@<\value{proofcount}%
    \advance\count@\@ne
  \repeat}
\makeatother

\newcommand{\Qopt}{Q^{\text{opt}}}
\newcommand{\qopt}{q^{\text{opt}}}

\newcommand{\poly}{\operatorname{poly}}

\newcommand{\E}{\mathbb{E}}

\newcommand{\PP}{\mathcal{P}}
\newcommand{\edit}{\Delta_{\text{ed}}}
\newcommand{\ed}{\Delta_{\text{ed}}}
\newcommand{\ham}{\operatorname{Ham}}
\newcommand{\polylog}{\operatorname{polylog}}

\newtheorem{thm}{Theorem}[section]
\newtheorem{defn}[thm]{Definition}
\newtheorem{lem}[thm]{Lemma}
\newtheorem{prop}[thm]{Proposition}

\newtheorem{cor}[thm]{Corollary}

\newtheorem{observation}[thm]{Observation}

\fixstatement{thm}
\fixstatement{lem}
\fixstatement{prop}
\fixstatement{cor}

\theoremstyle{definition}

\theoremstyle{remark}
\newtheorem{rem}[thm]{Remark}
\newtheorem{ex}[thm]{Example}

\theoremstyle{remark}
\numberwithin{equation}{section}

\begin{document}
\title{On Estimating Edit Distance: Alignment, Dimension Reduction, and Embeddings}
\author{Moses Charikar\thanks{Supported by NSF grant CCF-1617577 and a Simons Investigator Award.} \\
Stanford University\\
moses@cs.stanford.edu
\and
Ofir Geri\thanks{Supported by NSF grant CCF-1617577, a Simons Investigator Award, and the Google Graduate Fellowship in Computer Science in the School of Engineering at Stanford University.} \\
Stanford University\\
ofirgeri@cs.stanford.edu
\and
Michael P. Kim\thanks{Supported by NSF grant CCF-1763299.} \\
Stanford University\\
mpk@cs.stanford.edu
\and
William Kuszmaul\\
Stanford University\\
kuszmaul@cs.stanford.edu
}
\date{}
\maketitle

\begin{abstract}
Edit distance is a fundamental measure of distance between strings and has been widely studied in computer science. While the problem of estimating edit distance has been studied extensively, the equally important question of actually producing an alignment (i.e., the sequence of edits) has received far less attention. Somewhat surprisingly, we show that any algorithm to estimate edit distance can be used in a black-box fashion to produce an approximate alignment of strings, with modest loss in approximation factor and small loss in run time. Plugging in the result of Andoni, Krauthgamer, and Onak, we obtain an alignment that is a $(\log n)^{O(1/\varepsilon^2)}$ approximation in time $\tilde{O}(n^{1 + \varepsilon})$.

Closely related to the study of approximation algorithms is the study
of metric embeddings for edit distance. We show that min-hash techniques can be
useful in designing edit distance embeddings through three results:
(1) An embedding from Ulam distance (edit
distance over permutations) to Hamming space that matches the best known distortion of
$O(\log n)$ and also implicitly encodes a sequence of edits between
the strings; (2) In the case where the edit distance between the input strings is known to have an upper bound $K$, we show that embeddings of edit distance into Hamming space with distortion $f(n)$ can be modified in a black-box fashion to give distortion $O(f(\operatorname{poly}(K)))$ for a class of periodic-free strings; (3) A randomized dimension-reduction map with
contraction $c$ and asymptotically optimal expected distortion $O(c)$,
improving on the previous $\tilde{O}(c^{1 + 2 / \log \log \log n})$
distortion result of Batu, Ergun, and Sahinalp.
\end{abstract}

\section{Introduction}

The \emph{edit distance} $\ed(x, y)$ between two strings $x$ and $y$ is the
minimum number of character insertions, deletions, and substitutions
needed to transform $x$ into $y$.  This is a fundamental distance
measure on strings, extensively studied in computer science
\cite{WagnerF74,Embedding,ApproxSubPolyDistortion,ApproxPolyLog,UlamEmbedding,LowerBoundfirst,LowerBoundsecond,chakraborty2016streaming,backurs2015edit}.
Edit distance has applications in areas including computational
biology, signal processing, handwriting recognition, and image
compression \cite{navarro2001guided}. One of its oldest and most
important uses is as a tool for comparing differences between genetic
sequences \cite{navarro2001guided, blast, needleman1970general}.

The textbook dynamic-programming algorithm for edit distance runs in
time $O(n^2)$ \cite{WagnerF74, needleman1970general,
  vintsyuk1968speech}, and can be leveraged to recover a sequence of
edits, also known as an \emph{alignment}. The quadratic run time is
prohibitively large for massive datasets (e.g., genomic data), and
conditional lower bounds suggest that no strongly subquadratic time
algorithm exists \cite{backurs2015edit}.

The difficulty of computing edit distance has motivated the
development of fast heuristics \cite{gusfield1997algorithms,
  navarro2001guided, heuristicoriginal, blast}. On the theoretical
side, the tradeoff between run time and approximation factor (or
distortion) is an important question (see
\cite[Section~6]{indyk2001algorithmic}, and
\cite[Section~8.3.2]{indyk20048}). Andoni and Onak
\cite{ApproxSubPolyDistortion} (building on beautiful work of
Ostrovsky and Rabani \cite{Embedding}) gave an algorithm that
estimates edit distance within a factor of $2^{\tilde{O}(\sqrt{\log
    n})}$ in time $n^{1 + o(1)}$.  The current best known tradeoff was
obtained by Andoni, Krauthgamer and Onak \cite{ApproxPolyLog}, who
gave an algorithm that estimates edit distance to within factor $(\log
n)^{O(1/\varepsilon)}$ with run time $O(n^{1 + \varepsilon})$.

\paragraph*{Alignment Recovery}
While these algorithms produce estimates of edit distance, they do
not produce an alignment between strings (i.e., a sequence of
edits). By decoupling the problem of numerical estimation from the
problem of alignment recovery, the authors of
\cite{ApproxSubPolyDistortion} and \cite{ApproxPolyLog} are able to
exploit techniques such as metric space embeddings\footnote{The
  algorithm of \cite{ApproxSubPolyDistortion} has a recursive
  structure in which at each level of recursion, every substring
  $\alpha$ of some length $l$ is assigned a vector $v_\alpha$ such
  that the $\ell_1$ distance between vectors closely approximates edit
  distance between substrings. The vectors at each level of recursion
  are constructed from the vectors in lower levels through a series of
  procedures culminating in an application of Bourgain's embedding to
  a sparse graph metric. As a result, although the distances between
  vectors in the top level of recursion allow for a numerical
  estimation of edit distance, it is not immediately clear how one might attempt
  to extract additional information from the vectors in order to
  recover an alignment.
} and random sampling in order to obtain better approximations.
The algorithm of \cite{ApproxPolyLog} runs in phases, with the $i$-th phase
distinguishing between whether $\ed(x, y)$ is greater than or
significantly smaller than $\frac{n}{2^i}$. At the beginning of each
phase, a nuanced random process is used to select a small fraction of
the positions in $x$, and then the entire phase is performed while
examining only those positions. In total, the full algorithm samples
an $\tilde{O}\left(\frac{n^{\varepsilon}}{\ed(x, y)}\right)$-fraction of the letters in
$x$. Note that this is a polynomially small fraction as we are interested in the case where $\ed(x,y) > n^{1/2}$ (if the edit distance is small, we can run the algorithm of Landau et al.~\cite{landau1998incremental} in linear time).
Given that the algorithm only views a small portion of the
positions in $x$, it is not clear how to recover a global alignment
between the two strings.

We show, somewhat surprisingly, that any edit distance estimator can
be turned into an approximate aligner in a black box fashion with
modest loss in approximation factor and small loss in run time.  For
example, plugging the result of \cite{ApproxPolyLog} into our
framework, we get an algorithm with distortion $(\log
n)^{O(1/\varepsilon^2)}$ and run time $\tilde{O}(n^{1 + \varepsilon})$.  To the best
of our knowledge, the best previous result that gave an approximate
alignment was the work of Batu, Ergun,
and Sahinalp \cite{DimensionReduction}, which has distortion that is
polynomial in $n$.

\paragraph*{Embeddings of Edit Distance Using Min-Hash Techniques}
The study of approximation algorithms for edit distance closely
relates to the study of embeddings \cite{Embedding,
  ApproxSubPolyDistortion, DimensionReduction}. An \emph{embedding}
from a metric space $M_1$ to a metric space $M_2$ is a map
of points in $M_1$ to $M_2$ such that distances are preserved up to some
factor $D$, known as the \emph{distortion}. Loosely speaking,
low-distortion embeddings from a complex metric space $M_1$ to a
simpler metric space $M_2$ allow algorithm designers to focus on
the simpler metric space, rather than directly handling the more
complex one. Embeddings from edit distance to Hamming space have been
widely studied \cite{UlamEmbedding, Embedding,
  chakraborty2016streaming, belazzougui2016edit} and have played
pivotal roles in the development of approximation algorithms
\cite{ApproxSubPolyDistortion} and streaming algorithms
\cite{chakraborty2016streaming, belazzougui2016edit}.

The second contribution of this paper is to introduce new algorithms
for three problems related to embeddings for edit distance. The
algorithms are unified by the use of min-hash techniques to select
pivots in strings. We find this technique to be particularly useful
for edit distance because hashing the content of strings allows us to
split strings in places where their content is aligned, thereby
getting around the problem of insertions and deletions misaligning the
strings. In several of our results, this allows us to obtain
algorithms which are either more intuitive or simpler than their
predecessors. The three results are summarized below.

\textbf{Efficiently Embedding the Ulam Metric into Hamming Space: }For the
special case of the Ulam metric (edit distance on permutations), we
present a randomized embedding $\phi$ of permutations of size $n$ to
$\poly(n)$-dimensional Hamming space with distortion $O(\log n)$.  Given strings $x$ and $y$,
the Hamming differences between $\phi(x)$ and $\phi(y)$ not only
approximate the edit distance between $x$ and $y$, but also implicitly
encode a sequence of edits from $x$ to $y$. If the output string of
our embedding is stored using a sparse vector representation, then the
embedding can be computed in linear time, and its output can be stored
in linear space. The logarithmic distortion of our embedding matches
that of Charikar and Krauthgamer's embedding into
$\ell_1$-space \cite{UlamEmbedding}, which did not encode the actual edits and needed
quadratic time and number of dimensions. Our embedding also supports
efficient updates, and can be modified to reflect an edit in expected
time $O(\log n)$ (as opposed to the deterministic linear time required
by \cite{UlamEmbedding}).

\textbf{Embedding Edit Distance in the Low-Distance Regime: }Recently,
there has been considerable attention devoted to edit distance in the
low-distance regime
\cite{chakraborty2016streaming,belazzougui2016edit}.  In this regime, we are interested in finding algorithms that run faster or perform better given the promise that the edit distance between the input strings is small. This regime is
of considerable interest from the practical point of view.  Landau,
Myers and Schmidt \cite{landau1998incremental} gave an exact algorithm
for strings with edit distance $K$ that runs in time $O(n+K^2)$.  Recently,
Chakraborty, Goldenberg and Kouck{\`y} \cite{chakraborty2016streaming}
gave a randomized embedding of edit distance into Hamming space that has distortion linear in the edit distance with probability at least
$2/3$.

Given an embedding with distortion $\gamma(n)$ (a function of the
input size), could one obtain an embedding whose distortion is a
function of $K$, the edit distance, instead of $n$?  We answer this
question in the affirmative for the class of $(D,R)$-periodic free
strings. We say that a string is \emph{$(D, R)$}-periodic free if none
of its substrings of length $D$ are periodic with period of at most
$R$.  For $D \in \poly(K)$ and $R=O(K^3)$, we show that the embedding
of Ostrovksy and Rabani \cite{Embedding} can be used in a black-box
fashion to obtain an embedding with distortion $2^{O\left(\sqrt{\log K
    \log \log K}\right)}$ for $(D, R)$-periodic free strings with edit
distance of at most $K$. Our result can be seen as building on the
min-hash techniques of \cite[Section~3.5]{UlamEmbedding} (which in
turn extends ideas from \cite{bar2004approximating}). The authors of
\cite{UlamEmbedding} give an embedding for $(t, 180tK)$-non-repetitive
strings with distortion $O(t\log (tK))$ \cite{UlamEmbedding}.  The key
difference is that our notion of $(D,R)$-periodic free is much less
restrictive than the notion of non-repetitive strings studied in
\cite{UlamEmbedding}.

\textbf{Optimal Dimension Reduction for Edit Distance: }The
aforementioned work of Batu et al.~\cite{DimensionReduction} introduced and studied an interesting notion
of dimension reduction for edit distance: An embedding of edit distance
on length-$n$ strings to edit distance on length-$n/c$ strings (with
larger alphabet size) is called a \emph{dimension-reduction map} with
\emph{contraction} $c$. By first performing dimension reduction, one
can then apply inefficient algorithms to the contracted strings at a
relatively small overall cost. This idea was used in
\cite{DimensionReduction} to design an approximation algorithm with
approximation factor $O(n^{1/3 + o(1)})$. We provide a dimension-reduction map with
contraction $c$ and asymptotically optimal expected distortion
$O(c)$, improving on the distortion of $\tilde{O}(c^{1+2/\log \log
  \log n})$ obtained by the deterministic map of
\cite{DimensionReduction}.\footnote{When comparing these
  distortions, one should note that $2/\log \log \log n$ goes to zero
  very slowly; in particular, $c^{1+2/\log \log \log n} \ge c^{1.66}$
  for all $n \le 10^{82}$, the number of atoms in the universe.}

\section{Preliminaries}\label{secpreliminaries}

Throughout the paper, we will use $\Sigma$ to denote an
alphabet,\footnote{We assume that characters in $\Sigma$ can be
  represented in $\Theta(\log n)$ bits, where $n$ is the size of input
  strings.} and $\Sigma^n$ to denote the set of words of length $n$
over that alphabet. Additionally, we use $\PP_n$ to denote the set of
permutations of length $n$ over $\Sigma$, or equivalently, the subset
of $\Sigma^n$ containing words whose letters are distinct. When
referring to strings and permutations of length at most $n$, we use
$\Sigma^{\le n}$ and $\PP_{\le n}$ respectively. Given a string $w$ of length $n$, we denote its letters by $w_1, w_2,
\ldots, w_n$, and we use $w[i:j]$ to denote the substring $w_iw_{i +
  1} \cdots w_j$ (which is empty if $j < i$).

An \emph{edit operation} is either an insertion of a letter into a
word, a deletion of a letter from a word, or a substitution of a
letter with another letter. Given words $x$ and $y$, an
\emph{alignment} from $x$ to $y$ is a sequence of edits transforming
$x$ to $y$. The \emph{edit distance} $\ed(x, y)$ is the minimum number
number of edits needed to transform $x$ to $y$. Alternatively, it is
the length of an optimal alignment.

The following definition and lemma concerning periodic strings will
often be useful.
\begin{defn}
A string $w$ is \emph{periodic} with \emph{period} $p$ (or
\emph{$p$-periodic}) if $w_i = w_{i + p}$ for all $i \in \{1, \ldots,
|w| - p\}$.
\end{defn}

\begin{lem}
  Suppose $w \in \Sigma^n$ is $p$-periodic and $q$-periodic,
  with $n \ge p + q$. Then $w$ is $\gcd(p,
  q)$-periodic.
  \label{lemgcdperiodic}
\end{lem}
\proofatend
  Without loss of generality, assume $p \leq q$. For convenience, denote the letters of $w$ by $w_0, \ldots, w_{n - 1}$. Because $w$ is $q$-periodic and is of length at least $p + q$, we
  have that $w_i = w_{i + q}$ for all $i \in \{0, \ldots, p - 1\}$. It
  follows from the fact that $w$ is $p$-periodic that $w_i = w_{i + q}
  = w_{{(i + q)} \bmod p}$ for all $i \in \{0, \ldots, p - 1\}$. By
  the extended Euclidean algorithm, there exist integers $a$ and $b$ such that
  $\gcd(p, q) = ap + bq$, and thus $\gcd(p, q) \equiv bq \pmod
  p$. Repeatedly using the fact that $w_{i \bmod p} = w_i = w_{(i + q) \bmod p}$ for
  all $i \in \{0, \ldots, p - 1\}$, it follows that $w_i = w_{(i + bq)
    \bmod p} = w_{(i + \gcd(p, q)) \bmod p}$, and thus that the first $p$ letters of $w$ are $\gcd(p,
  q)$-periodic. Since $w$ is $p$-periodic, and the first $p$ letters
  of $w$ are $\gcd(p, q)$-periodic, we get that all of $w$ is $\gcd(p,
  w)$-periodic as well.
\endproofatend

A family $\mathcal{H}$ of hash functions from $S$ to $T$ is said to be $n$-wise independent for every distinct $x_1,\ldots,x_n \in S$, and every $y_1,\ldots,y_n \in T$ (not necessarily distinct), $\Pr[x_1 = y_1, \ldots, x_n = y_n] = \frac{1}{|T|^n}$.

Several of our algorithms assume access to a family of hash function
$\mathcal{H}$ from $\Theta(\log n)$ bits to $\Theta(\log n)$ bits such
that $h \in \mathcal{H}$ is $n$-wise independent and can be evaluated
in constant time. One can simulate this using the randomly generated hash family
$\mathcal{G}$ by Pagh and Pagh \cite{amazinghash} which is with high probability independent on any given set of $n$
elements $S$; although it
is a slight abuse of notation, we often refer to $\mathcal{G}$ as
being \emph{$n$-wise independent with high probability}. The family
$\mathcal{G}$ requires $O(n)$ preprocessing time and $O(n \log n)$ random bits,
but allows for constant-time evaluation.\footnote{Note that we use the
  model in which $O(\log n)$ random bits can be generated in constant
  time.}

Finally, when we have a pairwise independent hash function $h$, it will often
be useful to treat it as though whenever $x \neq y$, $\Pr[h(x) = h(y)]
= 0$. This is called the \emph{no-collisions assumption}, and the
following lemma shows that for a polynomial-time algorithm we can
assume it with high probability.
\begin{lem}
Let $S$ be a set of objects with $|S| \in \poly(n)$ and pick any $p
\in \poly(n)$. Then there exists a sufficiently large $b \in
\Theta(\log n)$ such that the following holds. If we select $h$ at
random from a pairwise independent family of hash functions mapping
$S$ to $b$ bits, then $h$ has probability at least $1 -
\frac{1}{p(n)}$ of being injective on $S$. That is, the no-collisions
assumption holds with probability at least $1 - \frac{1}{p(n)}$.
\label{lemnocollisions}
\end{lem}
\proofatend
Select $b \ge \log (|S|^2p) \in \Theta(\log n)$. Then for distinct $x,
y \in S$, $\Pr[h(x) = h(y)] \le \frac{1}{|S|^2p(n)}$. Applying the union
bound over all pairs $x, y \in S$, it follows that there is a
collision with probability at most $\frac{1}{p(n)}$.
\endproofatend

\section{Alignment Recovery Using a Black-Box Approximation Algorithm}\label{secblackbox}
In this section, we show how to transform a black-box algorithm
$\mathcal{A}$ for approximating edit distance into an algorithm
$\mathcal{B}$ (with different approximation ratio and run time) for
finding an approximately optimal alignment between strings.  The
algorithm $\mathcal{B}$ appears here as
Algorithm~\ref{alg:blackboxnew}. In the description of the algorithm, we
rely on the following definition of a partition.

\begin{defn}
  A \emph{partition} of a string $u$ into $m$ parts is a tuple $P =
  (p_0, p_1, p_2, \ldots, p_m)$ such that $p_0 = 0$, $p_m = |u|$, and
  $p_0 \le p_1 \le \cdots \le p_m$. For $i \in \{1, \ldots, m\}$, the
  \emph{$i$-th part of the partition $P$} is the subword $P_i :=
  u[p_{i - 1} + 1: p_i]$, which is considered to be empty if $p_i =
  p_{i - 1}$. A partition of a string $u$ into $m$ parts is an
  \emph{equipartition} if each of the parts is of size either $\lfloor
  |u|/m \rfloor$ or $\lceil |u|/m \rceil$.
\end{defn}

\begin{algorithm}
\caption{Black-Box Approximate Alignment Algorithm \label{alg:blackboxnew}}
Input: Strings $u, v$ with $|u| + |v| \le n$.

Parameters: $m \in \mathbb{N}$ satisfying $m \ge 2$ and an
approximation algorithm $\mathcal{A}$ for edit distance.

\begin{enumerate}
\item If $|u| \le 1$, then find an optimal alignment in time $O(|v|)$ naively.
\item Let $P = (p_0, p_1, \ldots, p_m)$ be an equipartition of $u$.
\item Let $S$ consist of the positions in $v$ which can be reached by
  adding or subtracting a power of $(1 + \frac{1}{m})$ to some
  $p_i$. Formally, define
  $$S = \left(\{p_0, \ldots, p_m\} \cup \{|v|\} \cup \left\{\left\lceil p_i \pm \left(1 +
  \frac{1}{m}\right)^j \right\rceil \mid i, j \ge 0\right\} \right)
  \cap \{0, \ldots, |v|\}.$$
\item Using dynamic programming, find a partition $Q=(q_0,\ldots,q_m)$ of
  $v$ such that each $q_i$ is in $S$, and such that the cost
  $\sum_{i=1}^{m}{\mathcal{A}(P_i,Q_i)}$ is minimized:
\begin{enumerate}
\item For $l \in S$, let $f(l,j)$ be the subproblem of returning a
  choice of $q_0, q_1, \ldots, q_{j}$ with $q_j = l$ which minimizes
  $\sum_{i = 1}^j \mathcal{A}(P_i, Q_i).$
\item Solve each $f(l, j)$ by examining precomputed answers for each
  of the subproblems of the form $f(l', j - 1)$ with $l' \le l \in S$:
  if $f(l', j - 1)$ gives a choice of $q_0, q_1, \ldots, q_{j - 1}$
  with $\sum_{i = 1}^{j - 1} \mathcal{A}(P_i, Q_i) = t$, then we can
  set $q_j = l$ to get $\sum_{i = 1}^j \mathcal{A}(P_i, Q_i) = t +
  \mathcal{A}(P_j, v[l' + 1: l])$. \\ (Here, $ \mathcal{A}(P_j, v[l' + 1:
    l])$ is computed using $\mathcal{A}$.)
\end{enumerate}
\item Recurse on each pair $(P_i,Q_i)$. Combine the resulting
  alignments between each $P_i$ and $Q_i$ to obtain an alignment
  between $u$ and $v$.
\end{enumerate}
\end{algorithm}

Formally, we assume that the approximation algorithm $\mathcal{A}$ has
the following properties:
\begin{enumerate}
\item There is some non-decreasing function $\gamma$ such
  that for all $n > 0$, and for any two strings $u, v$ with $|u| + |v|
  \le n$,
$$\edit(u, v) \le \mathcal{A}(u, v) \le \gamma(n) \cdot \edit(u, v)$$
\item $\mathcal{A}(u, v)$ runs in time at most $T(n)$ for some
  non-decreasing function $T$ which is super-additive in the sense that
  $T(j) + T(k) \le T(j + k)$ for $j, k \ge 0$.
\end{enumerate}

We are now ready to state the main theorem of this section.
\begin{thm}\label{thm:blackbox-main}
For all $u, v$ with $|u| + |v| \le n$ and $m \geq 2$, Algorithm~\ref{alg:blackboxnew}
outputs an alignment from $u$ to $v$ such that the number of edits is
at most $(3 \gamma(n))^{O(\log_m n)} \cdot \ed(u,v)$. Moreover, the
algorithm runs in time $\tilde{O}(m^5 \cdot T(n))$.
\end{thm}

Before continuing, we provide a brief discussion of Algorithm
\ref{alg:blackboxnew}. The algorithm first breaks $u$ into a partition
$P$ of $m$ equal parts. It then uses the black-box algorithm
$\mathcal{A}$ to search for a partition $Q$ of $v$ such that $\sum_i
\ed(P_i, Q_i)$ is near minimal; after finding such a $Q$, the
algorithm recurses to find approximate alignments between $P_i$ and
$Q_i$ for each $i$. Rather than considering every option for the
partition $Q = (q_0, \ldots, q_m)$, the algorithm limits itself to
those for which each $q_i$ comes from a relatively small set $S$.

The set $S$ is carefully designed so that although it is small (which
enables good run time), any optimal partition $\Qopt$ of $v$ can be in
some sense well approximated by some partition $Q$ using only $q_i$
values from $S$. The result is that the multiplicative error
introduced at each level of recursion will be bounded by $3\gamma(n)$;
across the $O(\log_m n)$ level of recursion, the total multiplicative
error will then be $3\gamma(n)^{O(\log_m n)}$. The fact that the
recursion depth appears in the exponent of the multiplicative error is
why we partition $u$ and $v$ into many parts at each level.

Next we discuss several implications of Theorem
\ref{thm:blackbox-main}. The parameter $m$ allows us to trade off the
approximation factor and the run time of the algorithm. When taken to
the extreme, this gives two particularly interesting results.

\begin{cor}
Let $0 < \varepsilon < 1$ (not necessarily constant). Then $m$ can be
chosen so that Algorithm \ref{alg:blackboxnew} has approximation ratio
$(3\gamma(n))^{O(\frac{1}{\varepsilon})}$ and run time
$\tilde{O}\left(T(n) \cdot n^{\varepsilon}\right)$.
\label{corsmalldistortion}
\end{cor}
\begin{proof}
Set $m = n^{\varepsilon/5}$. Then the approximation ratio of Algorithm~\ref{alg:blackboxnew} is
$(3 \gamma(n))^{O(\log_m n)} = (3 \gamma(n))^{O(\frac{1}{\varepsilon})}$
and the run time on $u, v$ with $|u| + |v| \le n$ is
$\tilde{O}(m^5 \cdot T(n)) = \tilde{O}\left(T(n) \cdot n^{\varepsilon}\right).$
\end{proof}

\begin{cor}
Let $0 < \varepsilon < 1$ (not necessarily constant). Then $m$ can be chosen so that Algorithm \ref{alg:blackboxnew}
has approximation ratio $n^{O(\varepsilon)}$ and run time
$\tilde{O}(T(n)) \cdot (3\gamma(n))^{O(1/\varepsilon)}$.
\label{corsmallruntime}
\end{cor}
\begin{proof}
  Set $m = (3 \gamma(n))^{\frac{1}{\varepsilon}}$. The approximation ratio is
  $(3 \gamma(n))^{O(\log_m n)} = (3 \gamma(n))^{O(\varepsilon \log_{3\gamma(n)} n)} = n^{O(\varepsilon)},$
  and the run time on $u, v$ with $|u| + |v| \le n$ is
$\tilde{O}(m^5 \cdot  T(n)) = \tilde{O}(T(n)) \cdot (3\gamma(n))^{O(1/\varepsilon)}.$
\end{proof}

Note that one can also apply our results to approximation algorithms
$\mathcal{A}$ which are randomized in the sense that they succeed with
probability $2/3$, meaning that $\Pr[\ed(x, y) \le \mathcal{A}(x, y)
  \le \gamma(n)\ed(x, y)] \ge \frac{2}{3}.$ In particular, given such
an algorithm $\mathcal{A}$, we can amplify the probability of success
to $1 - \frac{1}{\poly(n)}$ by performing $\Theta(\log n)$ independent
computations of $\mathcal{A}(x, y)$ and keeping the median return
value. Moreover, by the union bound, the amplified $\mathcal{A}$ will
succeed with probability $1 - \frac{1}{\poly(n)}$ for every pair $(x,
y)$ we run it on, as long as we invoke it only polynomially many
times. Hence we can apply our results to such algorithms $\mathcal{A}$
without modification, at the cost of an additional logarithmic factor
in run time and failure probability $\frac{1}{\poly(n)}$.

With this in mind, we can apply Corollary~\ref{corsmalldistortion} to
the randomized algorithm of Andoni et al.~\cite{ApproxPolyLog} with
approximation ratio $(\log n)^{O(1/\varepsilon)}$ and run time $O(n^{1 +
  \varepsilon})$, in order to obtain the following concrete result. (Note that $\varepsilon$ may be $o(1)$.)
\begin{cor}
  There exists an approximate-alignment algorithm which
  runs in time $\tilde{O}(n^{1 + \varepsilon})$, and has approximation
  factor $(\log n)^{O(1/\varepsilon^2)}$ with probability $1 -
  \frac{1}{\poly(n)}$.
  \label{corAKO}
\end{cor}

The remainder of the section is devoted to proving Theorem \ref{thm:blackbox-main}.

\subsection{Proof of Theorem~\ref{thm:blackbox-main}}

The proof of the theorem will follow from
Proposition~\ref{propblackboxruntime}, which bounds the run time of
Algorithm \ref{alg:blackboxnew}, and
Proposition~\ref{proplowapproxratio}, which bounds the approximation
ratio.

Throughout this section, let $u, v$ and $m$ be the values given to Algorithm
\ref{alg:blackboxnew}. Let $P = (p_0, \ldots, p_m)$ be the
equipartition of $u$ into $m$ parts, and let $S$ be the set defined by
Algorithm \ref{alg:blackboxnew}. We begin by bounding the run time of
Algorithm \ref{alg:blackboxnew}.

\begin{prop}
  Algorithm \ref{alg:blackboxnew} runs in time $\tilde{O}(T(|u| + |v|)
  \cdot m^5)$.
  \label{propblackboxruntime}
\end{prop}
\begin{proof}
  If $|u| \le 1$, then we can find an optimal alignment in time
  $O(|v|)$ naively.

  Suppose $|u| > 1$. Notice that $|S| \le O(m^2 \log n)$. In
  particular, because $(1 + \frac{1}{m})^{(m + 1) \ln n} \ge n$,
  \begin{equation*}
    S \subseteq \{p_0, \ldots, p_m\} \cup \{|v|\} \cup \left\{\left\lceil p_i \pm \left(1 +
  \frac{1}{m}\right)^j \right\rceil \mid i \in [0: m], j \in [0: (m + 1)\ln n]\right\},
  \end{equation*}
  which has size at most $O(m^2 \log n)$.

  Finding an equipartition of $u$ can be done in linear time, and
  constructing $S$ takes time $O(|S|) = \tilde{O}(m^2)$. In order to
  perform the fourth step which selects $Q$, we must compute $f(l, j)$
  for each $l \in S$ and $j \in [0: m]$. This results in $O(m|S|) \le
  O(m^3 \log n)$ subproblems. To evaluate $f(l, j)$, we must consider
  each $l' \in S$ satisfying $l' \le l$, and then compute the cost of
  $f(l', j - 1)$ plus $\mathcal{A}(P_j, v[l' + 1: l])$ (which takes
  time at most $T(|u| + |v|)$ to compute). Therefore, each $f(l, j)$
  is computed in time $O(|S| \cdot T(|u| + |v|)) \le \tilde{O}(T(|u| +
  |v|) \cdot m^2)$. Because there are $O(m \cdot |S|) =
  \tilde{O}(m^3)$ subproblems of the form $f(l, j)$, the total run
  time of the dynamic program is $\tilde{O}(T(|u| + |v|) \cdot m^5)$.

  So far we have shown that the first level of recursion takes time
  $\tilde{O}(T(|u| + |v|) m^5)$. We will now extend this to consider
  the $i$-th level of recursion for any $i$. The sum of the lengths of
  the inputs to Algorithm \ref{alg:blackboxnew} at a particular level
  of the recursion is at most $|u| + |v|$. Therefore, by the
  super-additivity of $T(n)$, it follows that the time spent in any
  given level of recursion is at most $\tilde{O}(T(|u| + |v|)
  m^5)$. Because each level of recursion reduces the sizes of the
  parts of $u$ by a factor of $\Omega(m)$, the number of levels is at
  most $O(\log_m n) \le O(\log n)$. Therefore, the run time is
  $\tilde{O}(T(|u| + |v|) \cdot m^5).$

\end{proof}

When discussing the approximation ratio of Algorithm
\ref{alg:blackboxnew}, it will be useful to have a notion of edit
distance between partitions of strings.
\begin{defn}
  Given two partitions $C = (c_0, \ldots, c_m)$ and $D = (d_0, \ldots,
  d_m)$ of strings $a$ and $b$ respectively, we define
  $\ed(C, D) := \sum_{i} \ed(C_i, D_i)$.
\end{defn}

In order to bound the approximation ratio of Algorithm
\ref{alg:blackboxnew}, we will introduce, for the sake of analysis, a
partition $\Qopt = (\qopt_0, \ldots, \qopt_m)$ of $v$ satisfying
$\ed(P, \Qopt) = \ed(u, v)$. Recall that $P$ is fixed, which allows us
to use it in the definition of $\Qopt$.

We claim that some partition $\Qopt$ satisfying $\ed(P, \Qopt) =
\ed(u, v)$ must exist. If $u$ and $v$ differed by only a single edit,
one could start from $P$ and explicitly define $\Qopt$ so that $\ed(P, \Qopt) =
1$ (by a case analysis of which type of edit was performed). It can then be shown by induction on the number of edits that, in general, we can obtain a partition $\Qopt$
satisfying $\ed(P, \Qopt) = \ed(u, v)$.

Our strategy for bounding the approximation ratio of Algorithm
\ref{alg:blackboxnew} will be to compare $\ed(P, Q)$ for the partition
$Q$ selected by our algorithm to $\ed(P, \Qopt)$. We do this through
three observations.

The first observation upper bounds $\ed(P, Q)$. Informally, it shows
that the cost in edit distance which Algorithm \ref{alg:blackboxnew}
pays for selecting $Q$ instead of $\Qopt$ is at most $2\sum_{i = 0}^m
|q_i - \qopt_i|$.

\begin{lem}
  Let $Q = (q_0, \ldots, q_m)$ be a partition of $v$. Then $$\ed(P, Q) \le \ed(u, v) + 2\sum_{i = 1}^m |q_i - \qopt_i|.$$
  \label{lemblackboxupper}
\end{lem}
\begin{proof}
  Observe that
  \begin{equation*}
    \ed(P, Q) \le \ed(P, \Qopt) + \ed(\Qopt, Q)  = \ed(u, v) + \sum_{i = 1}^m \ed(\Qopt_i, Q_i).
  \end{equation*}

  Because $Q$ and $\Qopt$ are both partitions of $v$, $\ed(Q_i, \Qopt_i) \le
  |q_{i - 1} - \qopt_{i - 1}| + |q_{i} - \qopt_{i}|$. In particular, $|q_{i -
    1} - \qopt_{i - 1}|$ insertions to the left side of one of $Q_i$
  or $\Qopt_i$ (whichever has its start point further to the right) will
  result in the two substrings having the same start-point; and then
  $|q_i - \qopt_i|$ insertions to the right side of one of $Q_i$ or
  $\Qopt_i$ (whichever has its end point further to the left) will result
  in the two substrings having the same end-point. Thus
  \begin{equation*}
    \begin{split}
      \ed(u, v) + \sum_{i = 1}^m \ed(\Qopt_i, Q_i) & \le \ed(u, v) + \sum_{i = 1}^m |q_{i - 1} - \qopt_{i - 1}| + |q_{i} - \qopt_{i}| \\
      & \le  \ed(u, v) + 2\sum_{i = 1}^m |q_i - \qopt_i|,
    \end{split}
  \end{equation*}
  where we are able to disregard the case of $i = 0$ because $q_0 =
  \qopt_0 = 0$.
\end{proof}

The next observation establishes a lower bound for $\ed(u, v)$.

\begin{lem}
  $\ed(u, v) \ge \frac{1}{m} \sum_{i = 1}^m |p_i - \qopt_i|.$
  \label{lemblackboxlower}
\end{lem}
\begin{proof}
  Because $\ed(P, \Qopt) = \ed(u, v)$, we must have that for each $i \in
  [m]$,
  $$\ed(u, v) = \ed(u[1: p_i], v[1: \qopt_i]) + \ed(u[p_i + 1: |u|], v[\qopt_i
    + 1: |v|]).$$ Notice, however, that the strings $u[1: p_i]$ and
  $v[1: \qopt_i]$ differ in length by at least $|\qopt_i - p_i|$. Therefore,
  their edit distance must be at least $|\qopt_i - p_i|$, implying that
  $\ed(u, v) \ge |\qopt_i - p_i|$.

  It follows that $\frac{1}{m} \ed(u, v) \ge \frac{1}{m}|\qopt_i -
  p_i|$. Summing over $i \in [m]$ gives the desired equation.
\end{proof}

So far we have shown that the cost in edit distance which Algorithm
\ref{alg:blackboxnew} pays for selecting $Q$ instead of $\Qopt$ is at
most $2\sum_{i = 0}^m |q_i - \qopt_i|$ (Lemma \ref{lemblackboxupper}),
and that the edit distance from $u$ to $v$ is at least $\frac{1}{m}
\sum_{i = 1}^m |p_i - \qopt_i|$ (Lemma \ref{lemblackboxlower}). Next
we compare these two quantities. In particular, we show that if $Q$ is
chosen to mimic $\Qopt$ as closely as possible, then each of the $|q_i
- \qopt_i|$ will become small relative to each of the $|p_i -
\qopt_i|$.

\begin{lem}
  There exists a partition $Q = (q_0, \ldots, q_m)$ of $v$ such that
  each $q_i$ is in $S$, and such that for each $i \in [0: m]$,
  $|q_i - \qopt_i| \le \frac{1}{m} |p_i - \qopt_i|.$
  \label{lempickgoodQ}
\end{lem}
\begin{proof}
  Consider the partition $Q$ in which $q_i$ is chosen to be the
  largest $s \in S$ satisfying $s \le \qopt_i$. Observe that: (1)
  because $0 \in S$, each $q_i$ always exists; (2) because $|v| \in S$,
  we will have $q_m = |v|$; (3) and because $\qopt_0 \le \qopt_1 \le
  \cdots \le \qopt_m$, we will have that $q_0 \le q_1 \le \cdots \le
  q_m$. Therefore, $Q$ is a well-defined partition of $v$.

  It remains to prove that $|q_i - \qopt_i| \le \frac{1}{m} |p_i -
  \qopt_i|$. We consider three cases:
  \begin{itemize}
  \item \textbf{Case 1: $\qopt_i = p_i$.} If $\qopt_i = p_i$, then $\qopt_i \in S$
    and thus $q_i = \qopt_i$ as well, meaning that $0 = |q_i - \qopt_i| \le
    \frac{1}{m} |p_i - \qopt_i| = 0$ trivially.
  \item \textbf{Case 2: $p_i < \qopt_i$.} Consider the largest
    non-negative integer $j$ such that $p_i + (1 + \frac{1}{m})^j \le
    \qopt_i$. By definition of $j$,
    \begin{equation}
      \left(1 + \frac{1}{m}\right)^{j} \le \qopt_i - p_i \le \left(1 + \frac{1}{m}\right)^{j + 1} .
      \label{eqjclose1}
    \end{equation}

    It follows that
    $$\qopt_i - \left(p_i + \left(1 + \frac{1}{m}\right)^j\right) \le \left(1 + \frac{1}{m}\right)^{j + 1} - \left(1 + \frac{1}{m}\right)^j.$$
    Simplifying, this becomes,
    \begin{equation}
      \qopt_i - \left(p_i + \left(1 +
      \frac{1}{m}\right)^j\right) \le \frac{1}{m}\left(1 + \frac{1}{m}\right)^j.
      \label{eqjclose2}
    \end{equation}
    Since $\lceil p_i + (1 + \frac{1}{m})^{j} \rceil \in S$, the
    definition of $q_i$ ensures that $q_i$ is between $p_i + (1 +
    \frac{1}{m})^{j}$ and $\qopt_i$ inclusive. Therefore,
    \eqref{eqjclose2} implies
    \begin{equation}
       \qopt_i - q_i \le \frac{1}{m}\left(1 + \frac{1}{m}\right)^j .
      \label{eqjclose3}
    \end{equation}
    Combining \eqref{eqjclose1} with \eqref{eqjclose3}, it follows
    that $\qopt_i - q_i \le \frac{1}{m} (\qopt_i - p_i)$, as desired.
  \item \textbf{Case 3: $p_i > \qopt_i$.} This case is similar to Case
    2. Consider the smallest positive integer $j$ such that $p_i - (1
    + \frac{1}{m})^j \le \qopt_i$. By definition of $j$,
    \begin{equation}
      \left(1 + \frac{1}{m}\right)^{j - 1} \le p_i - \qopt_i \le \left(1 + \frac{1}{m}\right)^{j}.
      \label{eqjclose1a}
    \end{equation}
    Manipulating this yields $\qopt_i - p_i \le -\left(1 + \frac{1}{m}\right)^{j - 1}$, which in turn implies
    $$\qopt_i -\left(p_i - \left(1 + \frac{1}{m}\right)^{j}\right) \le
    \left(1 + \frac{1}{m}\right)^{j} - \left(1 + \frac{1}{m}\right)^{j
      - 1}.$$ Simplifying, this becomes,
    \begin{equation}
     \qopt_i -\left(p_i - \left(1 + \frac{1}{m}\right)^{j}\right) \le
     \frac{1}{m}\left(1 + \frac{1}{m}\right)^{j - 1}.
      \label{eqjclose2a}
    \end{equation}
    Since $\max(\lceil p_i - (1 + \frac{1}{m})^{j} \rceil, 0) \in S$,
    the definition of $q_i$ ensures that $q_i$ is between $p_i - (1 +
    \frac{1}{m})^{j}$ and $\qopt_i$ inclusive. Therefore,
    \eqref{eqjclose2a} implies
    \begin{equation}
      \qopt_i - q_i \le \frac{1}{m}\left(1 + \frac{1}{m}\right)^{j - 1}.
      \label{eqjclose3a}
    \end{equation}
    Combining \eqref{eqjclose1a} with \eqref{eqjclose3a}, it follows
    that $\qopt_i - q_i \le \frac{1}{m} (p_i - \qopt_i)$, as desired.

  \end{itemize}
\end{proof}

We are now equipped to bound the approximation ratio of Algorithm
\ref{alg:blackboxnew}, thereby completing the proof of Theorem
\ref{thm:blackbox-main}. In particular, the preceding lemmas will
allow us to bound the approximation ratio at each level of recursion
to $O(\gamma(n))$. The approximation ratio will then multiply across
the $O(\log_m n)$ levels of recursion, giving total approximation
ratio $O(\gamma(n))^{O(\log_m n)}$.

\begin{prop}
 Let $E(u, v)$ be the number of edits returned by Algorithm
 \ref{alg:blackboxnew}. Then
 $$\ed(u, v) \le E(u, v) \le \ed(u, v) \cdot (3\gamma(n))^{O(\log_m
   n)}.$$
 \label{proplowapproxratio}
\end{prop}
\begin{proof}
  Because Algorithm \ref{alg:blackboxnew} finds a sequence of edits
  from $u$ to $v$, clearly $\ed(u, v) \le E(u,
  v)$.

  In order to show that $E(u, v) \le \ed(u, v) \cdot (3\gamma(n))^{O(\log_m n)}$,
  we use the previous lemmas. By Lemma \ref{lempickgoodQ} there is
  some partition $Q = (q_0, \ldots, q_m)$ of $v$ such that each $q_i$
  is in $S$, and such that for each $i \in [0: m]$,
  $|q_i - \qopt_i| \le \frac{1}{m} |p_i - \qopt_i|.$
  By Lemma \ref{lemblackboxlower}, it follows that
  $$\sum_{i = 1}^m |q_i - \qopt_i| \le \frac{1}{m} \sum_{i = 1}^m |p_i - \qopt_i| \le \ed(u, v).$$ Applying Lemma
  \ref{lemblackboxupper}, we then get that
  $$\ed(P, Q) \le \ed(u, v) + 2\sum_{i = 1}^m |q_i - \qopt_i| \le 3\ed(u, v).$$

  Thus there is some $Q$ which Algorithm \ref{alg:blackboxnew} is
  allowed to select such that $\ed(P, Q) \le 3\ed(u, v)$. Since the
  approximation ratio of $\mathcal{A}$ is $\gamma(n)$, the partition
  $Q$ which the algorithm actually chooses at the first level of
  recursion must therefore satisfy
  \begin{equation}
    \ed(P, Q) \le 3\gamma(n) \ed(u, v).
    \label{eqratioatsinglelevel}
  \end{equation}

  Recall from the proof of Proposition \ref{propblackboxruntime} that
  Algorithm \ref{alg:blackboxnew} has $O(\log_m n)$ levels of
  recursion. After the $i$-th level of recursion, $u$ has implicitly
  been split into a large partition $P^i$, $v$ has implicitly been
  split into a large partition $Q^i$, and the recursive subproblems
  are searching for edits between pairs of parts of $P^i$ and
  $Q^i$. Using \eqref{eqratioatsinglelevel}, we get by induction that $\ed(P^i, Q^i) \le (3\gamma(n))^{i}\ed(u,
  v)$. Since there are $O(\log_m n)$ levels of recursion, it follows
  that the number of edits returned by the algorithm is at most
  $\ed(u, v) \cdot (3 \gamma(n))^{O(\log_m n)}$.

\end{proof}

We conclude the section with two remarks.
\begin{rem}
  Because $m$ is typically asymptotically small, we have not made an
  effort to optimize the run time exponent of $m$. In particular, the
  run time of $\tilde{O}(T(n) m^5)$ can be improved to
  $\tilde{O}(T(n)m^3)$ by limiting each $q_i$ to be in the set
  $$\{p_i\} \cup \{|v|\} \cup \left\{\left\lceil p_i \pm \left(1 + \frac{1}{m}\right)^j\right\rceil \mid i,
  j \ge 0\right\}.$$ This comes at the expense of making it slightly more
  difficult to prove a variant of Lemma \ref{lempickgoodQ}, since it
  is no longer easy to ensure that the selection of $q_0, q_1 ,\ldots
  , q_m$ satisfies $q_0 \le q_1 \le \cdots \le q_m$.
\end{rem}

\begin{rem}
  If $\gamma(n)$ is $1 + o(1)$, then the 3 appearing in the
  approximation ratio $(3\gamma(n))^{O(\log_m n)}$
  becomes a bottleneck. This can be handled by redefining
  $$S := \left(\{p_0, \ldots, p_m\} \cup \{|v|\} \cup \left\{\left\lceil p_i \pm \left(1
  + \frac{1}{m^c}\right)^j \right\rceil \mid i, j \ge 0\right\}\right) \cap \{0, \ldots, |v|\}$$ for
  some constant $c > 1$ so that the options for each $q_i$ are at a
  finer granularity. This will achieve approximation ratio $\left((1 +
  O(1/m^{c - 1}))\gamma(n)\right)^{O(\log_m n)}$ at the expense of increasing $m$'s exponent in
  the run time.
\end{rem}

\section{Alignment Embeddings for Permutations}\label{secpermsintermsofn}

In this section we present a randomized embedding from $\PP_n$, the
set of permutations of length $n$, into Hamming space with expected
distortion $O(\log n)$. The embedding has the surprising property that
it implicitly encodes alignments between strings. Moreover, if the
output of the embedding is stored using run-length
encoding,\footnote{In run-length encoding, runs of identical
  characters are stored as a pair whose first entry is the character
  and the second entry is the length of the run.} then the size of the
output and the run time are both $O(n)$.

For convenience, in this subsection, we make the Simple Uniform
Hashing Assumption \cite{CLRS}, which allows us to assume a fully
independent family $\mathcal{H}$ of hash functions mapping
$\Theta(\log n)$ bits to $\Theta(\log n)$ bits with constant time
evaluation. This can be simulated using the family of
\cite{amazinghash}, which is independent on any given set of size $n$
with high probability, and which is further discussed in Section
\ref{secpreliminaries}.

\begin{algorithm}
\caption{Alignment Embedding for Permutations\label{alg:align-perm}}
Input: A string $w = w_1 \cdots w_n \in \PP_n$.

Parameters: $\varepsilon$ and $m \ge \log_{1/2 + \varepsilon} \frac{1}{n} + 1$.

\begin{enumerate}
\item At the first level of recursion only:
    \begin{enumerate}
    \item Initialize an array $A$ of size $2^m - 1$ (indexed starting at
    one) with zeros. The array $A$ will contain the output embedding.
    \item Select a hash function $h$ mapping $\Sigma$ to $r \log n$
    bits for a sufficiently large constant $r$.
    \end{enumerate}
\item Let $i$ minimize $h(w_i)$ out of the $i \in [n/2 - \varepsilon n:
  n/2 + \varepsilon n]$.\footnotemark{} We call $w_i$ the \emph{pivot} in $w$.
\item Set $A[2^{m - 1}] = w_i$.
\item Recursively embed $w_1 \cdots w_{i - 1}$ into $A[1:2^{m - 1} - 1]$.
\item Recursively embed $w_{i+1} \cdots w_{n}$ into $A[2^{m - 1} +
  1:2^m - 1]$.
\end{enumerate}
\end{algorithm}
\addtocounter{footnote}{-1}
\stepcounter{footnote}\footnotetext{With high probability, there are no hash collisions.}

The description of the embedding appears as Algorithm~\ref{alg:align-perm}. For simplicity, we assume $0 \notin \Sigma$, which allows us to use $0$ as a null character.
The algorithm takes two parameters: $\varepsilon$ and $m$. The parameter $\varepsilon$ controls a trade-off between the distortion and the output dimension. The parameter $m$ dictates the maximum depth of recursion that can be performed within the array $A$. In particular, $m$ needs to be chosen such that the algorithm does not run out of space for the embedding in the recursive calls.

Since each recursive step takes as input words of size in the
range $[(1/2 - \varepsilon) n: (1/2 + \varepsilon) n]$, the input size at the $i$-th level of the recursion is at most $(1/2 + \varepsilon)^{i-1} n$. We need to choose $m$ such that at the $m$-th level of recursion, the input size will be at most $1$.
Therefore, it suffices to pick $m$ satisfying
$$m \ge \log_{1/2 + \varepsilon} \frac{1}{n} + 1.$$

We denote the resulting embedding of the input string $w$ into the output array $A$ by $\phi_{\varepsilon, m}(w)$. Moreover, for $m = \lceil \log_{1/2 + \varepsilon} \frac{1}{n} + 1\rceil$, we define $\phi_\varepsilon(w)$ to be $\phi_{\varepsilon, m}(w)$. Note that
$\phi_\varepsilon$ embeds $w$ into an array $A$ of size
$$O\left(2^{\log_{1/2 + \varepsilon} 1/n}\right) = O\left(n^{-1/\log (1/2
  + \varepsilon)}\right),$$ which one can verify for $\varepsilon \le
\frac{1}{4}$ is $O(n^{1 + 6 \varepsilon})$.

We call $\phi_{\varepsilon}$ an \emph{alignment embedding} because
$\phi_{\varepsilon}$ maps a string $x$ to a copy of $x$ spread out across
an array of zeros. When we compare $\phi_{\varepsilon}(x)$ with
$\phi_{\varepsilon}(y)$ by Hamming differences, $\phi_{\varepsilon}$ encodes an alignment between $x$ and $y$; it pays for
every letter which it fails to match up with another copy of the same
letter. In particular, every pairing of a letter with a null
corresponds to an insertion or deletion, and every pairing of a
letter with a different letter corresponds to a substitution. As a
result, $\ham(\phi_\varepsilon(x), \phi_\varepsilon(y))$ will always be at
least $\edit(x, y)$.

The rest of this section is dedicated to prove the following theorem which summarizes the properties of $\phi_\varepsilon$.

\begin{thm}
  For $\varepsilon \le \frac{1}{4}$, there exists a randomized embedding
  $\phi_{\varepsilon}$ from $\PP_n$ to $O(n^{1 + 6
    \varepsilon})$-dimensional Hamming space with the following
  properties.
  \begin{itemize}
  \item For $x,y \in \PP_n$, $\phi_\varepsilon(x)$ and
    $\phi_{\varepsilon(y)}$ encode a sequence of
    $\ham(\phi_{\varepsilon}(x), \phi_{\varepsilon}(y))$ edits from $x$ to
    $y$. In particular, $\ham(\phi_{\varepsilon}(x), \phi_{\varepsilon}(y))
    \ge \ed(x, y)$.
  \item For $x, y \in \PP_n$, $\E[\ham(\phi_{\varepsilon}(x),
    \phi_{\varepsilon}(y))] \le O\left(\frac{1}{\varepsilon} \log n\right) \cdot \edit(x, y)$.
  \item For $x \in \PP_n$, $\phi_\varepsilon(x)$ is sparse in the sense
    that it only contains $n$ non-zero entries. Moreover, if
    $\phi_\varepsilon(x)$ is stored with run-length encoding, it can be
    computed in time $O(n)$.
  \end{itemize}
  \label{thmalignmentembedding}
\end{thm}

The first property in the theorem follows from the discussion above. In order to prove that $\E[\ham(\phi_\varepsilon(x), \phi_\varepsilon(y))] \le \edit(x, y)
O\left(\frac{1}{\varepsilon} \log n\right)$, we will consider a series of at most $2\ed(x,y)$ insertions or deletions that are used to transform $x$ into $y$. Each substitution operation can be emulated by an insertion and a deletion.

This results in a series of intermediate strings starting from $x$ and ending in $y$ that differ by one insertion or deletion. Moreover, note that by
  ordering deletions before insertions, each of the intermediate
  strings will still be a permutation. In the following key lemma, we will bound the expected Hamming distance between pairs of strings that differ by one insertion (or equivalently, one deletion). By the triangle inequality, we get the bound on $\E[\ham(\phi_\varepsilon(x), \phi_\varepsilon(y))]$.

\begin{lem}
 Let $x \in \PP_n$ be a permutation, and let $y$ be a permutation
 derived from $x$ by a single insertion. Let $0 < \varepsilon \le \frac{1}{4}$ and let $m$
 be large enough so that $\phi_{\varepsilon, m}$ is well-defined on $x$
 and $y$. Then $\E[\ham(\phi_{\varepsilon, m}(x), \phi_{\varepsilon, m}(y))]
 \le O\left(\frac{1}{\varepsilon} \log n\right)$.
 \label{lemembeddingexp}
\end{lem}
\begin{proof}

  Observe that the set of letters in position-range $[(1/2 -
    \varepsilon)|x|: (1/2 + \varepsilon)|x|]$ in $x$ differs by at most
  $O(1)$ elements from the set of letters in position-range $[(1/2 -
    \varepsilon)|y|: (1/2 + \varepsilon)|y|]$ in $y$. Thus with probability
  $1 - O(1/(\varepsilon n))$, there will be a letter $l$ in the overlap
  between the two ranges whose hash is smaller than that of any other
  letter in either of the two ranges.\footnote{Note by Lemma
    \ref{lemnocollisions}, the probability of any collisions is
    negligible.} In other words, the pivot in $x$ (i.e., the
  letter in the position range with minimum hash) will differ from the
  pivot in $y$ with probability $O(1/(\varepsilon n))$.

  Therefore,
  \begin{equation*}
    \begin{split}
      \E[\ham(\phi_{\varepsilon, m}(x), \phi_{\varepsilon, m}(y))] & = \Pr[\text{pivots differ}] \cdot \E[\ham(\phi_{\varepsilon, m}(x), \phi_{\varepsilon, m}(y)) \mid \text{ pivots differ}]\\
      &  \quad + \Pr[\text{pivots same}] \cdot \E[\ham(\phi_{\varepsilon, m}(x), \phi_{\varepsilon, m}(y)) \mid \text{ pivots same}] \\
      & \le O\left(\frac{1}{\varepsilon n}\right) \cdot \E[\ham(\phi_{\varepsilon, m}(x), \phi_{\varepsilon, m}(y)) \mid \text{ pivots differ}] \\
      & \quad + \E[\ham(\phi_{\varepsilon, m}(x), \phi_{\varepsilon, m}(y)) \mid \text{ pivots same}].
    \end{split}
  \end{equation*}

  In general, $\ham(\phi_{\varepsilon, m}(x), \phi_{\varepsilon, m}(y))$ cannot exceed $O(n)$. Thus
  \begin{equation*}
    \begin{split}
      \E[\ham(\phi_{\varepsilon, m}(x), \phi_{\varepsilon, m}(y))] &  \le O\left(\frac{1}{\varepsilon n}\right) \cdot O(n) + \E[\ham(\phi_{\varepsilon, m}(x), \phi_{\varepsilon, m}(y)) \mid \text{ pivots same}] \\
      & \le O\left(\frac{1}{\varepsilon}\right) + \E[\ham(\phi_{\varepsilon, m}(x), \phi_{\varepsilon, m}(y)) \mid \text{ pivots same}].
    \end{split}
  \end{equation*}

  If the pivot in $x$ is the same as in $y$, then the insertion must
  take place to either the left or the right of the pivot. Clearly
  $\phi_{\varepsilon, m}(x)$ and $\phi_{\varepsilon, m}(y)$ will agree on
  the side of the pivot in which the edit does not occur. Inductively
  applying our argument to the side on which the edit occurs, we incur
  a cost of $O(1/\varepsilon)$ once for each level in the recursion. The
  maximum depth of the recursion is $O\left(\log_{1/2 + \varepsilon}
  \frac{1}{n}\right) = O(\log n)$. This gives us
  \begin{equation*}
    \begin{split}
      \E[\ham(\phi_{\varepsilon, m}(x), \phi_{\varepsilon, m}(y))] & \le O\left(\frac{1}{\varepsilon}\right) \cdot \log n,
    \end{split}
  \end{equation*}
  as desired.
\end{proof}

It remains only to analyze the run time of computing the
embedding. Notice that for $x \in \PP_n$, $\phi_\varepsilon(x)$ is a
string with $n$ non-zero entries. Consequently, $\phi_\varepsilon(x)$ can
be stored in space $\Theta(n)$ if runs of zeros are stored using
run-length encoding. Moreover, if we store $\phi_\varepsilon(x)$ with
run-length encoding, we will show that we can compute
$\phi_\varepsilon(x)$ in time $O(n)$. Recall that the Range Minimum Query
problem can be solved with linear preprocessing time and constant
query time \cite{RangeMinimumQuery}. Consequently, with linear
preprocessing time, we can build a data structure which supports
constant-time queries that take a contiguous substring of $x$ and
return which letter has minimum hash. Using this, each recursive step
in the embedding can be performed in constant time. Since each such
step writes one letter to the output, there are only $n$ steps in
total, and thus the run time is $O(n)$.

\begin{rem}
Simulating a fully independent hash function $h$ can be done using
$O(n\log n)$ random bits \cite{amazinghash}. It turns out that if we
are willing to tolerate a run time of $O(n \log n)$, then this can be
reduced to $\polylog(n)$ bits, as follows.

Select hash functions $h_1, \ldots, h_m$ from the
$1/2$-min-wise independent family $\mathcal{H}$ of \cite{HashReference} mapping $\Sigma$ to $r \log
n$ bits for a sufficiently large constant $r$. Then modify Algorithm
\ref{alg:align-perm} to use $h_i$ at the $i$-th level of recursion. By
using a separate hash function at each level of recursion, we allow
ourselves to analyze the levels independently. Because each $h_i$ is
$1/2$-min-wise independent, the analysis then follows
similarly as in the proof of Lemma \ref{lemembeddingexp}.

Notice that each $h_i$ uses $O(\log n)$ random bits \cite{HashReference}, resulting in
$O(\log^2 n)$ random bits in total. Moreover, a naive implementation
of the algorithm yields run time $O(n \log n)$.
\end{rem}

\section{Embedding Periodic-Free Substrings in the Low Edit Distance Regime}\label{secperiodicfree}

We say that a string is \emph{$(D, R)$-periodic} if it is length at
least $D$ and is periodic with period at most $R$ (for $R \leq D$). A string is
\emph{$(D, R)$-periodic free} if it contains no contiguous $(D,
R)$-periodic substrings.

Suppose we have an embedding $\psi$ from edit distance in $\Sigma^n$
into Hamming space with subpolynomial distortion $\gamma(n)$, meaning
there is some value $T \ge 1$ such that for all $x, y \in \Sigma^n$, we have
$\ed(x, y) \le \frac{1}{T}\ham(\psi(x), \psi(y)) \le \gamma(n) \cdot
\ed(x, y)$. Our goal is to use such an embedding as a black box in
order to obtain a new embedding for the so-called low edit distance
regime. The new embedding, which would be parameterized by a value
$K$, would take any two strings $x, y \in \Sigma^n$ with $\ed(x, y)
\le K$ and map $x$ and $y$ to Hamming space with distortion
$\gamma'(K)$, a function of $K$ rather than a function of $n$.

In this section we make progress toward such an embedding with the
added constraints that our strings $x$ and $y$ are $(D, R)$-periodic
free for $D \in \poly(K)$ of our choice and $R \in O(K^3)$. Our
embedding $\alpha$ takes two such strings with $\ed(x, y) \le K$ and
maps them into Hamming space with distortion $\gamma(\poly(K))$. If we
select the black-box embedding $\psi$ to be the embedding
of Ostrovsky and Rabani \cite{Embedding}, then this gives distortion
$2^{O\left(\sqrt{\log K \log \log K}\right)}$.

The reason for restricting $x$ and $y$ to be $(D, R)$-periodic free is
somewhat subtle. Ultimately, it will boil down to the following lemma.
\begin{lem}
Suppose $x$ is a $(D, R)$-periodic free word of length at most $D +
R$. Then each of the $D$-letter substrings of $x$ are distinct.
  \label{lemperiodicnew}
\end{lem}
\begin{proof}
  Suppose that two $D$-letter substrings $a = x_i\cdots x_{i + D - 1}$
  and $b = x_j \cdots x_{j + D - 1}$ of $x$ are equal with $i <
  j$. Then for $t \in \{1, \ldots, D\}$, since the $t$-th letter of
  $a$ equals the $t$-th letter of $b$, we have that $x_{i + t - 1} =
  x_{j + t - 1}$. It follows that $a$ is $(j - i)$-periodic (by definition). Because
  $a$ and $b$ are $D$-letter substrings and $x$ is length $D + R$, we
  have that $j - i \le R$, meaning that $a$ is $(D, R)$-periodic, a
  contradiction.
\end{proof}

The main step in our embedding is to partition the strings $x$ and
$y$ into parts of length $\poly(K)$ in a way so that the sum of the
edit distances between the parts equals $\ed(x, y)$ (with some
probability), and then to apply the black-box embedding $\psi$ to each
individual part. This step of the embedding is referred to as the \emph{primary embedding}. The primary embedding succeeds with some probability, but its distortion is not well bounded in expectation. By applying techniques from  \cite{UlamEmbedding}, we are able to take any high-probability low-distance-regime embedding, and turn it into an embedding with good expected distortion (Theorem \ref{thmprobtoexpected}), thereby completing the full construction of our embedding. The following theorem states the main result in this section.

\begin{thm}
  Suppose we have an embedding $\psi$ from edit distance in $\Sigma^n$
  to Hamming space with subpolynomial distortion $\gamma(n)$ (and with
  $\gamma(n) \ge 2$ for all $n$). Let $K \in \mathbb{N}$ and pick some
  $D \in \poly(K)$. Then there exists $R \in O(K^3)$ and an embedding
  $\alpha: \Sigma^n \rightarrow \{0, 1\}$ such that for $(D,
  R)$-periodic free $x, y \in \Sigma^n$ with $\ed(x, y) \le K$, we have
  $$\frac{1}{\gamma(O(K^3D))} \Omega\left(\ed(x, y)\right) \le K
  \Pr[\alpha(x) \neq \alpha(y)] \le O(\ed(x, y)).$$ In other words,
  $\alpha$ is an embedding from edit distance (at most $K$) between
  $(D, R)$-periodic free strings to scaled Hamming distance with
  expected distortion at most $O(\gamma(O(K^3D)))$.
  \label{thmloweditregime}
\end{thm}

The following result follows by plugging in the embedding of Ostrovsky and
Rabani \cite{Embedding},
which is defined over binary strings, and maps edit distance to
$\ell_1$-distance with distortion $2^{O\left(\sqrt{\log n \log \log
    n}\right)}$.
\begin{cor}
 Let $K \in \mathbb{N}$ and pick some $D \in \poly(K)$. Then there
 exists $R \in O(K^3)$ and an embedding $\alpha: \{0, 1\}^n \rightarrow
 \{0, 1\}$ such that for $(D, R)$-periodic free $x, y \in \{0, 1\}^n$ with $\ed(x, y) \le K$,
 we have
  $$\frac{1}{2^{O\left(\sqrt{\log K \log \log K}\right)}}
 \Omega\left(\ed(x, y)\right) \le K \Pr[\alpha(x) \neq \alpha(y)] \le
 O(\ed(x, y)).$$ In other words, $\alpha$ is an embedding from edit
 distance (at most $K$) to scaled Hamming distance with expected distortion
 $2^{O\left(\sqrt{\log K \log \log K}\right)}$.
 \label{corpluginOR}
\end{cor}
\begin{proof}
  Although the embedding of Ostrovsky and Rabani embeds into $\ell_1$
  rather than into Hamming space, it can be converted using standard
  techniques to an embedding into Hamming space with additional blow-up in dimension and constant additional
  multiplicative distortion.\footnote{In particular, if $x \in \mathbb{R}^\ell$ is the output of the embedding of \cite{Embedding}, then one can transform it as follows. First by shifting each coordinate's value, one can assume that the values of coordinates are always between $0$ and $O(n^2)$, since if the embedding ever had two outputs differing by more than $\Omega(n^2)$ in some coordinate, then the distortion of the embedding would be at least $\Omega(n)$. Then after rescaling up by $\ell$ and rounding each of the coordinates, one can presume each coordinate to be a positive integer of magnitude at most $O(n^2 \ell)$, while incurring at most constant additional distortion. Each
    coordinate can then be encoded in $O(n^2\ell)$ Hamming coordinates by
    mapping the value $i$ to the adjacently-placed values $1, 1, \ldots, 1, 0, 0,
    \ldots, 0$, where the $1$ is repeated $i$ times.} Therefore, applying Theorem
  \ref{thmloweditregime} to the embedding of Ostrovsky and Rabani, we
  get $\alpha$ satisfying
    $$\frac{1}{2^{O\left(\sqrt{\log O(K^3D) \log \log O(K^3D)}\right)}} \Omega\left(\ed(x, y)\right) \le K
  \Pr[\alpha(x) \neq \alpha(y)] \le O(\ed(x, y)).$$
  Since $D \in \poly(K)$, it follows that
    $$\frac{1}{2^{O\left(\sqrt{\log K \log \log K}\right)}} \Omega\left(\ed(x, y)\right) \le K
  \Pr[\alpha(x) \neq \alpha(y)] \le O(\ed(x, y)),$$
  as desired.
\end{proof}

The primary embedding is given by Algorithm \ref{alg:low-dist-embedding}. The algorithm uses as parameters a window-size $W$, a sub-window size $W'$, and a sub-sub-window
size $W''$. We will use the term \emph{window} to refer to $W$
consecutive letters in a word, the term \emph{sub-window} to refer to
$W'$ consecutive letters in a window, and the term
\emph{sub-sub-window} to refer to $W''$ consecutive letters within a
sub-window. We will assign values to $W, W', W''$ at the end of our
analysis.

Algorithm \ref{alg:low-dist-embedding} is designed so that when two strings $x$ and $y$ differ by a small number of edits, they are likely to be broken into similar window sequences. In particular, suppose that $A_t(x)$ and $A_t(y)$ are aligned with each other, meaning that some optimal sequence of edits from $x$ to $y$ maps the first letter of $A_t(x)$ to the first letter of $A_t(y)$ without editing the letter itself. Then consider what happens when we construct $A_{t + 1}(x)$ and $A_{t + 1}(y)$. By picking a sub-window uniformly at random, we guarantee that with high probability (as a function of $W,W'$) no edits occur directly within that sub-window. However, if $x$ and $y$ differ by an insertion or deletion
prior to that sub-window, the sub-window within $x$ may be misaligned
with the sub-window within $y$. Nonetheless, the set of $W''$-letter
sub-sub-windows of the sub-window of $x$ will be almost the same as
the set of $W''$-letter sub-sub-windows of the sub-window of $y$. By selecting the sub-sub-window with minimum hash as the
start position for the next window, we are then able to guarantee that
with high probability (as a function of $W'$) we pick start positions for $A_{t + 1}(x)$ and $A_{t + 1}(y)$ which are
aligned with each other. Note that in order for the min-hash technique to work, one needs the sub-sub-windows of any given sub-window to all be distinct; for appropriately selected $W, W', W''$, this will be a consequence of Lemma \ref{lemperiodicnew}.

Before formally analyzing Algorithm \ref{alg:low-dist-embedding}, we discuss its use of hash functions. For each $t \in \{1, \ldots, \frac{W}{2W'}\}$, the algorithm requires a $2n$-wise independent hash function $h_t$ mapping $\Sigma^{W''}$ to $\Theta(\log{n})$ bits. Note that the $h_t$ hash functions can be efficiently simulated using the family of \cite{amazinghash}. In particular, one can first use a pairwise independent hash function $g$ to map elements of $\Sigma^{W''}$ to $\Theta(\log n)$ bits (while avoiding collisions with high probability). Then, one can then select a function $h$ mapping $\Theta(\log n)$ bits to $\Theta(\log n)$ bits from the family of \cite{amazinghash}. Finally one can then define $h_t(u) = h((t, g(u))$, where $(t, g(u))$ represents the tuple containing $t$ and $g(u)$. For strings $x, y \in \Sigma^n$, and for each $t$, with high probability $h_t$ will be independent on the $W''$-letter substrings of $x$ and $y$. By the union bound, the independence of $h_t$ holds simultaneously for all $t$ with high probability.

\begin{algorithm}[t]
\caption{The Primary Embedding\label{alg:low-dist-embedding}}
Given:
(i) an input string $w$ of length at most $n$,
(ii) parameters $W'' \ll W' \ll W$,
(iii) an embedding $\psi$ from edit distance into Hamming space, where the input strings are of length at most $W$, and the output has a fixed length,
(iv) $s_1, \ldots, s_{2n / W}$ randomly selected elements of $\{1, \ldots, \frac{W}{2W'}\}$, and
(v) for each $t \in \{1, \ldots, \frac{W}{2W'}\}$, a $2n$-wise independent hash function $h_t$ mapping $\Sigma^{W''}$ to $\Theta(\log{n})$ bits.

\begin{enumerate}
\item Let $b = b_1 \cdots b_{2W}$ consist of dummy letters not in $\Sigma$. Define $w'$ to be the concatenation $wb$.
\item Construct a window sequence $A(w) = A_1(w), \ldots, A_{2n / W}(w)$ of windows in $w'$, by first setting $A_1(w)$ to consist of the first $W$ letters $w'_1 \cdots w'_W$ of $w'$, and then constructing each $A_{t + 1}(w)$ from $A_t(w)$ as follows:
\begin{enumerate}
\item If $A_{t}(w)$ consists entirely of dummy letters, then define $A_{t +  1}(w) = A_t(w)$. Otherwise, continue to the next step.
\item Divide the second half of the window $A_t(w)$ into $\frac{W}{2W'}$ non-overlapping sub-windows of size $W'$, and consider the $s_t$-th such sub-window.
\item Each substring of length $W''$ inside the sub-window is called a sub-sub-window (sub-sub-windows may overlap). Compute the hash $h_t(v)$ of each sub-sub-window $v$.
\item The window $A_{t+1}(w)$ will consist of $W$ letters starting at the same position as the sub-sub-window with the smallest hash. (Ties can be broken arbitrarily but, it turns out, will not occur with high probability.)
\end{enumerate}
\item Define $P(w)$ to be a partition of $w$ where the $i$-th part starts at the first letter of $A_i(w)$ and ends prior to the first letter of $A_{i+1}(w)$, excluding any dummy letters.
\item Embed each part of $P(w)$ using $\psi$, and concatenate the embedding results. The dimension of the concatenation is $\ell \cdot |P(w)|$ where $\ell$ denotes the output dimension of $\psi$. In order to ensure that the outputs of our algorithm are of a fixed dimension, pad the concatenated result with zeros as needed to bring its length to $\ell \cdot \lceil 2n / W \rceil$. Return the padded string.
\end{enumerate}
\end{algorithm}

We now present a formal analysis of Algorithm \ref{alg:low-dist-embedding}. For simplicity, we will ignore floors
and ceilings when convenient. We start with two definitions that will aid us in our discussion of the algorithm.

\begin{defn}
  Let $x$ and $y$ be words over $\Sigma \cup \{b_1, b_2, \ldots\}$. Fix some optimal sequence
  of edits between $x$ and $y$. Then a letter in $x$ or $y$ is
  said to be \emph{touched} if it is involved in an edit, or if an insertion or deletion occurs immediately to its right in the sequence of edits.
  A letter $a$ in $x$ is said to be \emph{siblings} with a letter $b$ in
  $y$ if either $a$ and $b$ are the first letters of $x$ and $y$, or
  the edits map $a$ to $b$ while leaving both untouched.
\end{defn}

\begin{ex}
Consider $x = abcdefghi$ and $y = bpdeqrsfghia$. Moreover, consider the following sequence of edits:
\[
\begin{array}{c c c c c c c c c c c c c}
a & b & c & d & e &   &   &   & f & g & h & i &  \\
  & b & p & d & e & q & r & s & f & g & h & i & a
\end{array}
\]
That is, we delete and reinsert the $a$, we substitute the $c$ with a $p$, and we insert the $q$, $r$, and $s$.

The touched letters in $x$ are $a$, $c$, $e$, and $i$. The touched letters in $y$ are $p$, $e$, $q$, $r$, $s$, $i$, and $a$. Between $x$ and $y$, the pairs of siblings are $(b, b)$, $(d, d)$, $(f, f)$, $(g, g)$, and $(h, h)$.
\end{ex}

In general, if there is a sequence of adjacent untouched letters in $x$, then the siblings of those letters will form the same adjacent sequence in $y$. This property of touched letters is the main motivation for the above definition.

\begin{defn}
  Let $P$ and $Q$ be partitions of strings $x$ and $y$ into $j$ parts
  $P_1, \ldots, P_j$ and $Q_1, \ldots, Q_j$ respectively.  We call
  $(P, Q)$ \emph{edit-preserving} if there exists an optimal sequence
  of edits from $x$ to $y$ such that the first letter of $P_i$ is
  siblings with the first letter of $Q_i$ for each $i$.
\end{defn}

Intuitively, if $(P, Q)$ is edit-preserving, then we can focus on the parts of $P$ and $Q$ separately when trying to embed $x$ and $y$ into Hamming space. This is formalized by the following lemma.

\begin{lem}
  Let $P$ and $Q$ be partitions of strings $x$ and $y$ into $j$ parts
  $P_1, \ldots, P_j$ and $Q_1, \ldots, Q_j$ respectively. If $(P, Q)$
  is edit-preserving, then
  $$\edit(x, y) = \sum_{i = 1}^j \edit(P_i, Q_i).$$
  \label{lempartitionnew}
\end{lem}
\begin{proof}
  Consider an optimal sequence of edits from $x$ to $y$ such that the
  first letter of $P_i$ is siblings with the first letter of $Q_i$ for
  each $i$.  We say that an edit occurs within $P_1$ if the first
  letter of $P_2$ is to the right of the edit. For $i > 1$, we say
  that an edit occurs within $P_i$ if it takes place between the first
  letter of $P_i$ and the first letter of $P_{i + 1}$ (or the end of
  the word, if $i = j$). Notice that because the first letter of each
  $P_i$ is siblings with the first letter of each $Q_i$, the edits
  occurring within $P_i$ transform $P_i$ into $Q_i$. It follows that
  $$\edit(x, y) = \sum_{i = 1}^j \edit(P_i, Q_i).$$
  as desired.
\end{proof}

For the rest of the section, for any string $w$, we use $P(w)$ to denote the partition constructed in Algorithm \ref{alg:low-dist-embedding}. The following key lemma establishes that under certain conditions, $(P(x), P(y))$ is
edit-preserving with high probability.
\begin{lem}
  Let $x, y \in \Sigma^n$ be words such that $\edit(x, y) \le
  K$. Moreover, suppose that for each sub-window of either $x$ or $y$,
  no two of its sub-sub-windows are equal. Fix a minimal sequence of
  edits between $x$ and $y$. Then $(P(x), P(y))$ is edit-preserving
  with probability at least
  $$1 - 16K \cdot \left(\frac{W'}{W} - \frac{1}{W' - W'' + 1}\right).$$
  \label{lemeditpreservingprobnew}
\end{lem}
\begin{proof}
Notice that, as long as the window $A_t(x)$ is not entirely contained
within the dummy letters $b$, then the window $A_{t + 1}(x)$ is
constructed to overlap it by at most $W / 2$ letters. It follows that
each letter $l$ in $x$ appears in at most two windows in $x$'s window
sequence. For each $A_t(x)$ which is not contained within the dummy
letters $b$, define $B_t(x)$ to be the sub-window of $A_t(x)$ used to
determine $A_{t + 1}(x)$'s starting point. Recall that $B_t(x)$ is
chosen uniformly at random (independently for each $t$) out of the sub-windows making
up the second half of $A_t(x)$. Thus a given letter in the second half of $A_t$ has
probability $\frac{W'}{W/2}$ of appearing within
$B_t(x)$. Since each letter appears in the second half of at most one window, a given
letter in $x$ has probability at most $\frac{2W'}{W}$ of appearing in
any $B_t(x)$. There are at most $2K$ touched letters in $x$. By the
union bound, the probability of any touched letter in $x$ appearing in
any $B_t(x)$ is at most $\frac{4KW'}{W}$.

Define $B_t(y)$ similarly as to $B_t(x)$. We shall assume for the rest
of the proof that none of the letters in any $B_t(x)$ or $B_t(y)$ are
touched by the edits. We additionally condition on each one of the $h_t$ functions having no collisions (on the elements we hash), which happens with probability $1-\frac{1}{\poly(n)}$ (by Lemma~\ref{lemnocollisions}).
By the analysis provided so far, these conditions will occur
with probability at least
\begin{equation}
  1 - \frac{8KW'}{W} -\frac{1}{\poly(n)} \geq 1 - \frac{16KW'}{W}.
  \label{probsubwindowbrokennew}
\end{equation}

Consider some window $A_t(x)$ (with $t < 2n/W$)
and the corresponding window $A_t(y)$. Suppose that the first letters
of $A_t(x)$ and $A_t(y)$ are siblings. In other words, either $t = 1$,
or the edits between $x$ and $y$ transform the first letter of
$A_t(x)$ to become the first letter of $A_t(y)$, without ever touching
either letter. We will show that
\begin{multline}
  \Pr\Big[\text{First letters of }A_{t + 1}(x) \text{ and }A_{t +
      1}(y)\text{ siblings} \mid \\ \text{First letters of }A_{t}(x)
    \text{ and }A_{t}(y)\text{ siblings} \Big] \ge 1 - \frac{2k_t}{W'
    - W'' + 1}.
  \label{eqsiblingsnew}
\end{multline}
where $k_t$ is defined to be the number of touched letters contained in either
$A_t(x)$ or $A_t(y)$. Also recall that these probabilities are conditioned on lack of hash collisions and on the event that touched letters do not appear in any $B_t(x)$ or $B_t(y)$.

The easy case for \eqref{eqsiblingsnew} occurs when either of $A_t(x)$ or
$A_t(y)$ consists entirely of dummy letters. Because the first letters
of $A_t(x)$ and $A_t(y)$ are siblings, it follows that both consist entirely of dummy letters, and $A_t(x) = A_t(y)$. Therefore,
with probability $1$, the first letter of $A_{t + 1}(x)$ will be
siblings with the first letter of $A_{t + 1}(y)$.

Suppose, on the other hand, that neither $A_t(x)$ nor $A_t(y)$
consists entirely of dummy letters. Then, because the sub-window
$B_t(x)$ of $A_t(x)$ consists of untouched letters, the same
sub-window must appear in $y$ (we will denote the image of $B_t(x)$ in
$y$ by $B'_t(x)$). Since $A_t(x)$ and $A_t(y)$ start with sibling letters, the offset (due to insertions or deletions) of the sub-window $B_t(x)$ in $A_t(x)$
must be within $k_t$ of the offset of $B'_t(x)$ in $A_t(y)$.
Since $B_t(y)$ appears at the same offset (determined by $s_t$) in $A_t(y)$ as does $B_t(x)$
in $A_t(x)$, it follows that $B_t(y)$ must overlap $B'_t(x)$ in at
least $W' - k_t$ letters. This means that, for each one of $B_t(x)$ and
$B_t(y)$, out of the $W' - W'' + 1$ sub-sub-windows, all but at most $k_t$ of
them appear in the overlap between $B_t(x)$ and $B_t(y)$. It follows that with probability at least $1 -
\frac{2k_t}{W' - W'' + 1}$, there is a unique sub-sub-window $u$ in
the overlap of $B_t(y)$ and $B'_t(x)$ whose $h_t$-hash is smaller than
that of any other sub-sub-window in $B_t(y)$ or $B'_t(x)$. This, in
turn, guarantees that the first elements of $A_{t + 1}(x)$ and $A_{t +
  1}(y)$ are siblings, completing the proof of \eqref{eqsiblingsnew}.

By repeatedly applying \eqref{eqsiblingsnew},
\begin{equation*}
  \begin{split}
    \Pr\Big[\text{First letters of }A_{t}(x) \text{ and }A_{t}(y)\text{ siblings }\forall t\Big] & \ge  \prod_{t = 1}^{2n/W - 1} \left(1 - \frac{2k_t}{W' - W'' + 1}\right) \\
     & \ge  1 - \sum_{t = 1}^{2n/W - 1} \left(\frac{2k_t}{W' - W'' + 1}\right).
  \end{split}
\end{equation*}

Since each of $x$ and $y$ contain at most $2K$ touched letters. Each touched letter in $x$ or $y$ appears in at most two windows of $A(x)$
or $A(y)$, respectively. It follows that $\sum_{t = 1}^{2n/W - 1}{k_t} \leq 8K$, and
\begin{equation}
  \Pr\Big[\text{First letters of }A_{t}(x) \text{ and }A_{t}(y)\text{ siblings }\forall t\Big] \ge  1 - \frac{16K}{W' - W'' + 1}.
  \label{eqallsiblingsnew}
\end{equation}

By definition, if the first letters of $A_t(x)$ and $A_t(y)$ are siblings
for all $t$, then $(P(x), P(y))$ will be edit preserving. Thus
\eqref{probsubwindowbrokennew} and \eqref{eqallsiblingsnew} combine to tell
us that $(P(x), P(y))$ are edit-preserving with probability at least
  $$1 - 16K \cdot \left(\frac{W'}{W} - \frac{1}{W' - W'' + 1}\right).$$
\end{proof}

The following theorem shows that with high
probability, Algorithm~\ref{alg:low-dist-embedding} behaves well on $(D, R)$-periodic free strings of distance
at most $K$ from each other.

\begin{thm}
  Suppose we have an embedding $\psi$ from edit distance in $\Sigma^n$
  to Hamming space with distortion $\gamma(n)$, meaning
there is some value $T \ge 1$ such that for all $x, y \in \Sigma^n$, we have
$\ed(x, y) \le \frac{1}{T}\ham(\psi(x), \psi(y)) \le \gamma(n) \cdot
\ed(x, y)$. Let $K, D, R, C$ be
  positive variables satisfying $R \ge 32KC$. Then there exists a
  randomized embedding $\phi$ from $\Sigma^n$ to Hamming space
  satisfying the following property. If $x, y \in \Sigma^n$ are $(D,
  R)$-periodic free and $\edit(x, y) \le K$, then
  $$\edit(x, y) \le \frac{1}{T}\ham(\phi(x), \phi(y)) \le \edit(x, y) \cdot
  \gamma(64 DKC),$$ with probability at least $1 - \frac{1}{C}$.
  \label{thmloweditregimewithhighprob}
\end{thm}
\begin{proof}
  Note that if $n < D$, then we can achieve $\gamma(D)$ distortion
  using the embedding $\psi$, and that if $n \ge D$ but $D < R$, then
  there are no $(D, R)$-periodic free strings $x, y \in \Sigma^n$,
  again making the result trivial. Thus we can assume that $R \le D
  \le n$.

  Select $W'' = D$, $W' = W'' + 32KC$ and $W = W' \cdot 32KC$. Let
  $\phi(x)$ and $\phi(y)$ be the outputs of Algorithm \ref{alg:low-dist-embedding}.

  Because $x$ and $y$ are $(D, R)$-periodic free, Lemma
  \ref{lemperiodicnew} tells us that for any substring $a$ of either
  $x$ or $y$ satisfying $|a| \le D + R$, the $D$-letter substrings of
  $a$ are distinct. Since $W' = D + 32KC \le D + R$ (from the statement of the theorem) and $W'' = D$, it
  follows that for each sub-window of $x$ and $y$, the sub-sub windows
  are distinct. Hence we can apply Lemma
  \ref{lemeditpreservingprobnew} to see that $(P(x), P(y))$ is
  edit-preserving with probability at least

  \begin{equation*}
    \begin{split}
      1 - 16K \cdot \left(\frac{W'}{W} - \frac{1}{W' - W'' + 1}\right)
      & = 1 - 16K \cdot \left(\frac{1}{32KC} + \frac{1}{32KC +
        1}\right) \\ & > 1 - \frac{1}{C}. \\
    \end{split}
  \end{equation*}

  Supposing that $(P(x), P(y))$ is edit-preserving, Lemma
  \ref{lempartitionnew} tells us that the distortion of $\phi$ on the
  edit distance between $x$ and $y$ is at most the distortion of the
  embedding $\alpha$ on words of length $W = (D + 32KC) \cdot
  32KC$. Since $32KC \le R \le D$, we have that $(D + 32KC) \cdot
  32KC \le 64DKC$. It follows that

  \begin{equation*}
    \begin{split}
      \edit(x, y) \le \frac{1}{T}\ham(\phi(x), \phi(y)) & \le \edit(x, y) \cdot
      \gamma(64 DKC),
    \end{split}
  \end{equation*}
  as desired.
\end{proof}

So far we have obtained an embedding which behaves well with
probability $1 - \frac{1}{C}$.  The next result, which is implicitly
present in Section 3.5 of \cite{UlamEmbedding} takes any high-probability
low-edit-distance-regime embedding from edit distance to Hamming
space, and turns it into an embedding with good expected distortion.

\begin{thm}
  \label{thmprobtoexpected}
Let $S$ be a subset of $\Sigma^{\le n}$. Suppose we have a randomized
embedding $\phi:S \rightarrow \Sigma^m$ such that for all $x, y \in S$
satisfying $\ed(x, y) \le K$,
$$\ed(x, y) \le \frac{1}{T}\ham(\phi(x), \phi(y)) \le F(K) \ed(x, y)$$
with probability at least $1 - \frac{1}{K \cdot F(K)}$ for some
distortion function $F(K)$ with $F(K) \ge 2$, for some $T \ge 1$,
and for all $K \in \mathbb{N}$.

Then there exists $\alpha:S \rightarrow \{0, 1\}$ such that for all
$x, y \in S$ satisfying $\ed(x, y) \le K$,
$$\Omega\left(\ed(x, y)\right) \le K \cdot F(K) \Pr[\alpha(x) \neq
  \alpha(y)] \le O(\ed(x, y)F(K)).$$ In other words, $\alpha$ is an
embedding from edit distance (at most $K$) over $S$ to Hamming distance (scaled) with expected distortion at
most $O(F(K))$.
\end{thm}
\begin{proof}
  Implicit in \cite{UlamEmbedding}. The full proof is deferred to Appendix~\ref{secimplicitresultofCandK}.
\end{proof}

Combining Theorems \ref{thmloweditregimewithhighprob} and
\ref{thmprobtoexpected}, we can now prove the main result for
this section.

\begin{proof}[Proof of Theorem~\ref{thmloweditregime}]
  Because $\gamma$ is subpolynomial, we can pick $C \in O(K^2)$ and $R
  = 32KC \in O(K^3)$ such that $K \cdot \gamma(64DKC) \le C$. Because
  $R \ge 32KC$, we can apply Theorem
  \ref{thmloweditregimewithhighprob} to obtain an embedding with
  distortion $\gamma(64 DKC)$ with probability $1 -
  \frac{1}{C}$. Because $K \cdot \gamma(64DKC) \le C$, we can then
  apply Theorem \ref{thmprobtoexpected} to obtain a map $\alpha:
  \Sigma^n \rightarrow \{0, 1\}$ as described.
\end{proof}

\begin{rem}
  The embedding described by Theorem \ref{thmloweditregime} maps
  strings to only a single bit. By concatenating together multiple
  independent iterations of the embedding, one can obtain a map to
  multiple bits for which the output distances are bounded with
  high probability in addition to just in expectation. In fact,
  although stated as a randomized result, the above embedding be
  iterated over all possible choices of random bits in order to give a
  deterministic embedding from edit distance
  to scaled Hamming space for $(D, R)$-periodic free strings whose
  distances are upper bounded by $K$.
\end{rem}

\begin{rem}
  The embedding described by Theorem \ref{thmloweditregime} can also
  be performed on pairs of $(D, R)$-periodic strings $x, y$ with $\ed(x,
  y) > K$. In this case, the embedding will encode the fact that
  $\ed(x, y) \ge \Omega(K)$. In particular, the embedding of Theorem
  \ref{thmloweditregimewithhighprob} will output strings of Hamming
  distance at least $T \cdot K$; and one can check that the embedding of
  Theorem \ref{thmprobtoexpected} will then satisfy
  $\frac{\Omega(K)}{\gamma(O(K^3D))} \le K \Pr[\alpha(x) \neq
    \alpha(y)]$.\footnote{This follows by the same logic as in Case 2
    of the proof of Theorem \ref{thmprobtoexpected}.}
\end{rem}

\section{Dimension Reduction}\label{secdimred}

A mapping of edit distance on length-$n$ strings to edit distance on
strings of length at most $n/c$ (with larger alphabet size) is called
a dimension-reduction map with contraction $c$. In this section we
present a randomized dimension-reduction map for edit distance whose
contraction is $c$ and whose distortion is $O(c)$, which is within a
constant factor of optimal. This marks an improvement over the
previous state of the art \cite{DimensionReduction}, which achieved
distortion $\tilde{O}\left(c^{1+\frac{2}{\log \log \log
    n}}\right)$. Moreover, our dimension-reduction map can be computed
in time $\tilde{O}(n)$.

We begin with a formal definition of a dimension-reduction map.

\begin{defn}
We say that a randomized map $\phi$ from strings in $\Sigma^n$ to
strings of length at most $n/c$ over a different alphabet $\Sigma'$ is
a dimension-reduction map with \emph{contraction} $c$ and
\emph{expected distortion} at most $\alpha \cdot \beta$ if for all $x,
y \in \Sigma^n$,
\begin{equation}
  \ed(x, y) \le \alpha \ed(\phi(x), \phi(y)),
  \label{eqexpecteddistortioneq1}
\end{equation}
and
\begin{equation}
  \E[\ed(\phi(x), \phi(y))] \le \beta \ed(x, y).
  \label{eqexpecteddistortioneq2}
\end{equation}
\end{defn}

Note, in particular, that the lower bound
\eqref{eqexpecteddistortioneq1} is required to hold regardless of the
random bits used to compute $\phi$. This is similar to the way
distortion was defined in \cite{DefnExpectedDistortion}.

Before continuing, we present a lower bound for the expected
distortion of a dimension-reduction map $\phi$. This is a simple
generalization of the analogous lower bound for deterministic
dimension reduction provided in \cite{DimensionReduction}.
\begin{lem}
  Suppose $\phi$ is a randomized dimension-reduction map with
  contraction $c$. Then the expected distortion of $\phi$ is at least
  $c$.
\end{lem}
\begin{proof}
  Consider $x, y \in \Sigma^n$ with $\ed(x, y) = 1$. In order so that
  \eqref{eqexpecteddistortioneq1} holds for some $\alpha$, it must be
  that $\ed(\phi(x), \phi(y)) \neq 0$. Therefore, $\ed(\phi(x),
  \phi(y)) \ge 1$, meaning that $\beta \ge 1$.

  Now consider $x, y \in \Sigma^n$ with $\ed(x, y) = n$. Since
  $\phi(x), \phi(y)$ are of length at most $n/c$, it follows that
  $\ed(\phi(x), \phi(y)) \le n/c$. Therefore, from
  \eqref{eqexpecteddistortioneq1}, $\alpha$ is at least $c$.

  Since $\alpha \ge c$ and $\beta \ge 1$, the expected distortion of
  $\phi$ is $\alpha \beta \ge c$.
\end{proof}

The main result in this section will be constructing a
dimension-reduction map $\phi$ for which the following theorem
holds.
\begin{thm}
  Let $c, n \in \mathbb{N}$. There exists a dimension-reduction map
  $\phi$ such that for every two strings $x, y \in \Sigma^{\le n}$,
  \begin{itemize}
  \item The lengths of $\phi(x)$ and $\phi(y)$ are at most $O(|x|/c)$ and $O(|y|/c)$, respectively.
  \item $\ed(x, y) \le 2c \cdot \ed(\phi(x), \phi(y))$.
  \item The expected distortion of $\phi$ is $O(c)$. In
    particular, $$\Pr[\ed(\phi(x),\phi(y)) > m\ed(x, y)] \leq
    \left(\frac{1}{2}\right)^{\Omega(m)} + \frac{1}{\poly(n)},$$ for
    all $m \in \mathbb{N}$ and for a polynomial $\poly(n)$ of our choice.
  \item Each of $\phi(x)$ and $\phi(y)$ can be computed in time
    $O(n\log c)$.
  \end{itemize}
  \label{thm:main-dim-reduction-embedding}
\end{thm}

The rest of the section proceeds as follows. In Subsection
\ref{secdimredforperms}, we present a simplified version of Theorem
\ref{thm:main-dim-reduction-embedding} for permutations. This
motivates our more general construction, and yields a linear-time
$O(\log n)$-approximation algorithm for finding an optimal sequence of
edits between permutations. Then in Subsection \ref{secdimredprelims},
we present preliminaries for the more general dimension-reduction map
$\phi$. In Subsection \ref{subsecdimred}, we define $\phi$ and prove
Theorem \ref{thm:main-dim-reduction-embedding}. Finally, in Subsection
\ref{secdimredapproxalg} we discuss an application of our
dimension-reduction map to approximation algorithms for edit distance.

\subsection{Dimension Reduction for Permutations}\label{secdimredforperms}

In this section we present a dimension-reduction map for
permutations. This motivates the techniques which we will use in
general, and has the interesting property that the map can be computed
in strictly linear time.

For convenience, in this subsection we assume access to a fully
independent family $\mathcal{H}$ of hash functions mapping
$\Theta(\log n)$ bits to $\Theta(\log n)$ bits with constant time
evaluation. This can be simulated using the family of
\cite{amazinghash}, which is independent on any given set of size $n$
with high probability, and which is further discussed in Section
\ref{secpreliminaries}.

\begin{algorithm}\caption{Dimension Reduction for Permutations with Edit Distance\label{alg:permembedding}}
Input: a string $w \in \PP_{\le n}$, a parameter $c \in \mathbb{N}$ with $c$
even, and a constant $b$.

\begin{enumerate}
\item Select $h$ at random from a family of $n$-wise independent hash
  functions mapping $\Sigma$ to $\{0, 1\}^{b \log n}$.
\item Call a letter $w_i$ a \emph{marker} if $c < i \le n - c$ and if
  $h(w_i)$ is smaller than any of $h(w_j)$ for $j \in \{i - c, \ldots,
  i + c\} \setminus \{i\}$. That is, $w_i$ is a marker if $h(w_i)$ is
  the unique minimum in the sequence $h(w_{i - c}), \ldots, h(w_{i +
    c})$.
\item Partition $w$ into blocks in the following way: Each marker
  indicates the beginning of a block. For each block of size at least
  $c$, subdivide it into blocks such that all of them are of size
  exactly $c$, except for the last one which will be of size at most
  $c$.
\item Return a string over the alphabet $\Sigma^*$ (i.e., whose
  letters are strings in $\Sigma$) where each block from the preceding
  step is a letter.
\end{enumerate}
\end{algorithm}

Given $w \in \PP_{\le n}$, Algorithm \ref{alg:permembedding} shows how
to compute $\phi(w)$. The embedding selects a letter $w_i$ to be a
marker if $w_i$'s hash is smaller than that of any of the preceding
$c$ letters or following $c$ letters. The algorithm then splits $w$
based on the markers, and then further splits each of the resulting
chunks into blocks of size $c$, except that the last block in each
chunk may be of size between $1$ and $c$.

The following theorem establishes that $\phi$ is a dimension-reduction
map.

\begin{thm}
Let $p \in \poly(n)$. Then there exists a sufficiently large constant
$b$ such that for every two permutations $x, y \in \PP_{\le n}$,
  \begin{itemize}
  \item The lengths of $\phi(x)$ and $\phi(y)$ are at most $O(|x|/c)$
    and $O(|y|/c)$, respectively.
  \item $\ed(x, y) \le c \cdot \ed(\phi(x), \phi(y))$.
  \item The distortion of $\phi$ is $O(c)$. In
    particular, $$\Pr[\ed(\phi(x),\phi(y)) > m\ed(x, y)] \leq
    \left(\frac{1}{2}\right)^{\Omega(m)} + \frac{1}{p(n)},$$ for all
    $m \in \mathbb{N}$.
  \item Each of $\phi(x)$ and $\phi(y)$ can be computed in time
    $O(n)$.
  \end{itemize}
    \label{thmdimredforperms}
\end{thm}

We prove each of the parts of the theorem individually. The first part
of Theorem \ref{thmdimredforperms} is established by the following
lemma.

\begin{lem}
  Let $w \in \PP_{\le n}$. Then for any contiguous $m \cdot c$ letters in
  $w$, $\phi(w)$ breaks the letters into at most $O(m)$ blocks.
  \label{lemshrinkperm}
\end{lem}
\begin{proof}
  Notice that no window of $c + 1$ letters in $w$ can contain two
  distinct markers $w_i$ and $w_j$. Indeed, in order for $w_i$ to be a
  marker, we would need that $h(w_i) < h(w_j)$. But then in order for
  $w_j$ to be a marker, we would also need that $h(w_j) < h(w_i)$, a
  contradiction.

  It follows that the total number of markers in the $mc$ letters is
  at most $O(m)$. Recall that once the markers are placed, the string
  is partitioned into blocks of size $c$, except that the block
  preceding a marker may be smaller. Thus in the $mc$ letters there
  are at most $O(m)$ blocks of length $c$ and then at most one
  additional block for each marker, totaling to $O(m)$ blocks.
\end{proof}

 The second part of Theorem \ref{thmdimredforperms} is established by
 the following lemma.
\begin{lem}\label{lem:distortion-lower}
  For $x, y \in \PP_{\le n}$,
  $$\ed(x, y) \le c \cdot \ed(\phi(x), \phi(y)).$$
\end{lem}
\begin{proof}
  Consider a sequence of edits from $\phi(x)$ to $\phi(y)$. Since the
  blocks in $\phi(x)$ and $\phi(y)$ are all of size at most $c$, each
  edit to $\phi(x)$ corresponds with a sequence of at most $c$ edits
  to $x$. Therefore, we get that $\ed(x, y) \le c \cdot \ed(\phi(x),
  \phi(y))$.
\end{proof}

In order to prove the third part of Theorem \ref{thmdimredforperms},
which bounds the distortion of $\phi$, we must first upper bound the
probability of having many consecutive letters without any markers. In
the following analysis we will condition on $h$ being injective (i.e.,
collision free). Later, we will apply the no-collisions assumption
(Lemma \ref{lemnocollisions}) to justify this with high probability.

\begin{lem}
Consider $w \in \PP_{\le n}$, $k \in \{1, \ldots, b\}$, and $m \in
\mathbb{N}$. Consider some $3cm$ consecutive letters, and condition on
$h$ being injective on those letters. Then the algorithm places a marker due
to $h$ in the $3cm$ letters with probability at least $1 -
\left(\frac{2}{3}\right)^{m}$.
\label{lemmarkerscommonperm}
\end{lem}
\begin{proof}
Consider a sequence of $3c$ consecutive characters $w_i \cdots w_{i +
  3c - 1}$ in $w$. With probability $\frac{1}{3}$, the letter $w_j$
with minimum hash will be in the middle $c$ out of the $3c$ letters;
that is, we will have $i + c \le j < i + 2c$. Should this occur, then
$w_j$ will be a marker, since its hash will be smaller than any of the
$c$ letters to its left or right. Hence with probability at least
$\frac{1}{3}$, the sequence of $3c$ consecutive letters will contain a
marker.

Now consider $3cm$ consecutive characters of $w$. If we break these
characters into $m$ disjoint chunks of $3c$ letters, then the argument
above demonstrates that the probability of the no chunk containing a
marker is at most $\left(\frac{2}{3}\right)^m$ (by independence of
$h$).
\end{proof}

In order to bound the distortion of $\phi$, we begin by considering the
case where $x$ and $y$ differ by a single insertion.
\begin{lem}
 Let $m \in \mathbb{N}$. Let $x \in \PP_{\le n}$ be a permutation and
 $y \in \PP_{\le n}$ be a permutation that is derived from $x$ by a
 single insertion.  Condition on the fact that $h$ is injective on the
 letters in $x$. Then, $\Pr[\ed(\phi(x),\phi(y)) > m] \leq
 \left(\frac{1}{2}\right)^{\Omega(m)}$.
 \label{lemphiaftersingleinsertionperm}
\end{lem}
\begin{proof}
Let $i$ be the index such that $y$ is generated from $x$ by adding a
character at index $i$. That is, $x_1 \cdots x_{i-1} = y_1 \cdots
y_{i-1}$ and $x_i x_{i+1} \cdots = y_{i+1} y_{i+2} \cdots$.

Whether a letter is a marker is determined entirely by the $c$ letters
immediately to its left and right. Thus for any $j < i - c$, there is
a marker at $x_j$ if and only if there is a marker at $y_j$, and for
any $j \ge i + c$, there is a marker at $x_j$ if and only if there is
a marker at $y_{j+1}$.

It follows that the blocks containing only characters prior to $x_{i -
  c- 1}$ and $y_{i - c - 1}$ are entirely identical in $\phi(x)$ and
$\phi(y)$. This is because, in general, any edit which modifies only
markers more than one character to the right of a block will not
affect that block.

Moreover, it follows that if there is a marker in some position $x_j$
with $j \ge i + c$, then there will also be a marker in $y_{j + 1}$,
and the blocks starting from $x_j x_{j + 1} \cdots$ will be identical
to those starting from $y_{j + 1}y_{j + 2}\cdots$. By Lemma
\ref{lemmarkerscommonperm}, with probability at least $1 -
\left(\frac{2}{3}\right)^m$, there is a marker in the $3mc$ letters
following $x_{i + c - 1}$. Consequently with probability at least $1 -
\left(\frac{2}{3}\right)^m$, the blocks in $x_j x_{j + 1} \cdots$ will
be identical to those in $y_{j + 1}y_{j + 2}\cdots$ for some $j < i +
3mc + c$.

So far we have shown that with probability at least $1 -
\left(\frac{2}{3}\right)^m$, the blocks in $x$ and $y$ are the same
except for those intersecting $x_{i - c - 1}, \ldots, x_{i + 3mc + c - 1}$
or $y_{i - c - 1}, \ldots, y_{i + 3mc + c}$. By Lemma
\ref{lemshrinkperm}, it follows that the number of edits needed to
transform $\phi(x)$ to $\phi(y)$ is at most $O(m)$, completing the
proof.
\end{proof}

The preceding lemma can be used to bound $\ed(\phi(x), \phi(y))$ with
high probability given that $x$ and $y$ differ by a single
insertion. In order to extend the lemma to hold for many edits
simultaneously, we must first make an observation about sums of
dependent geometric random variables.
\begin{lem}
  Let $X_1, X_2, \ldots, X_k$ be (not necessarily independent) random
  variables over $\mathbb{N}$ satisfying $\Pr[X_i > m] \le
  \frac{1}{2^m}$ for all $i$ and $m \ge 1$. Then for all $\lambda \in
  \mathbb{N}$,
  $$\Pr \left[\sum_{i = 1}^k X_i \ge \lambda k \right]  \le \frac{1}{2^{\Omega(\lambda)}}.$$
  \label{lemgeometricvars}
\end{lem}
\proofatend
  Let $t= 1.5^{1/k}$. Then,
  \begin{equation*}
    \begin{split}
      \Pr \left[ \sum_{i = 1}^k X_i \ge \lambda k \right] & = \Pr\left[ t^{\sum_i X_i} \ge t^{\lambda k} \right] \\
      & \le \frac{\E\left[t^{\sum_i X_i}\right]}{t^{\lambda k}},
    \end{split}
  \end{equation*}
  where the second inequality follows from Markov's inequality. Since
  $t^k = 1.5$, it follows that
  \begin{equation*}
    \Pr \left[ \sum_{i = 1}^k X_i \ge \lambda k \right] \le
    \frac{\E\left[t^{\sum_i X_i}\right]}{2^{\Omega(\lambda)}}.
  \end{equation*}

  To complete the proof, it suffices to show that $\E\left[t^{\sum_i
      X_i}\right] \le O(1)$. By H\"{o}lder's inequality,
  \begin{equation*}
    \E\left[t^{X_1} \cdot t^{X_2} \cdots t^{X_k}\right] \le \E\left[t^{X_1 k}\right]^{1/k} \cdot \E\left[t^{X_2 k}\right]^{1/k} \cdots \E\left[t^{X_k  k}\right]^{1/k}.
  \end{equation*}
  Without loss of generality $\E\left[t^{X_i k} \right]$ is maximized
  by $i = 1$. Thus the above expression is at most
  \begin{equation*}
    \begin{split}
      \E\left[t^{X_1 k}\right] & = \E\left[1.5^{X_1}\right] \\
      & \le \sum_{i = 1}^\infty 1.5^i \cdot \Pr[X_1 \ge i] \\
      & \le 1.5 + \sum_{i = 2}^\infty 1.5^i \cdot \frac{1}{2^{i - 1}} \\
      & = O(1),
    \end{split}
  \end{equation*}
  completing the proof.
\endproofatend

In fact, we will need a slightly specialized application of the above
lemma.
\begin{lem}
  Let $p$ be some (small) probability. Let $X_1, X_2, \ldots, X_k$ be
  (not necessarily independent) random variables over $\mathbb{N}$
  satisfying $\Pr[X_i > \lambda] \le \frac{1}{2^{\Omega(\lambda)}} +
  \frac{1}{p}$ for all $i$ and $\lambda \in \mathbb{N}$. Then for all $\lambda
  \in \mathbb{N}$,
  $$\Pr \left[\sum_{i = 1}^k X_i > \lambda k \right]  \le \frac{1}{2^{\Omega(\lambda)}} + \frac{2k}{p}.$$
  \label{lemgeometricvars2}
\end{lem}
\proofatend
  Because $\Pr[X_i > \lambda] \le \frac{1}{2^{\Omega(\lambda)}} +
  \frac{1}{p}$, there exists a positive constant $g$ such that for
  each $i$ if we define $Y_i = \lceil X_i / g \rceil$ then we get for $\lambda \in \mathbb{N}$ that
  \begin{equation}
    \Pr[Y_i > \lambda] \le
    \frac{1}{2}\left(\frac{1}{2}\right)^\lambda + \frac{1}{p}.
    \label{eqYibound}
  \end{equation}
    Thus for $\lambda
  \le \log p - 1$, we have that $\Pr[Y_i > \lambda] \le
  \left(\frac{1}{2}\right)^\lambda$. Define $Z_i = \min(Y_i, \lceil\log p\rceil -
  1)$. Then $Z_i$ satisfies that $\Pr[Z_i > \lambda] \le
  \left(\frac{1}{2}\right)^\lambda$ for all $\lambda \in \mathbb{N}$. Therefore,
  by Lemma \ref{lemgeometricvars}, we get that
  \begin{equation}
    \Pr \left[\sum_{i = 1}^k Z_i \ge \lambda k\right] \le \left(\frac{1}{2}\right)^{\Omega(\lambda)}.
    \label{eqZibound}
  \end{equation}
  For a given $i$, $$\Pr[Z_i \neq Y_i] = \Pr[Y_i > \lceil\log p\rceil - 1] \le
  \frac{2}{p},$$ where the last inequality comes from
  \eqref{eqYibound}. Thus by the union bound, the probability that
  $Z_i \neq Y_i$ for any $i$, is at most $\frac{2k}{p}$. Combining
  this with Equation \eqref{eqZibound}, we get that
  \begin{equation*}
    \Pr \left[\sum_{i = 1}^k Y_i > \lambda k\right] \le \left(\frac{1}{2}\right)^{\Omega(\lambda)} + \frac{2k}{p}.
    \label{eqYibound2}
  \end{equation*}
  If we recall that $Y_i$ and $X_i$ differ by a constant factor, it follows that
  \begin{equation*}
    \Pr \left[\sum_{i = 1}^k X_i > \lambda k\right] \le \left(\frac{1}{2}\right)^{\Omega(\lambda)} + \frac{2k}{p}.
    \label{eqXibound}
  \end{equation*}
\endproofatend

We are now prepared to bound the distortion of our dimension-reduction
map $\phi$.
\begin{lem}
  Select $p \in\poly(n)$, and pick $b$ to be a sufficiently large
  constant. Let $x$ and $y$ be permutations in $\PP_{\le n}$. Then for
  $m \in \mathbb{N}$,
  $$\Pr[\ed(\phi(x),\phi(y)) > m\ed(x, y)] \leq
   \left(\frac{1}{2}\right)^{\Omega(m)} + \frac{1}{p(n)}.$$
  \label{lemboundexpansionperm}
\end{lem}
\begin{proof}
  There is a sequence of $2\ed(x, y) + 1$ permutations $x^0, x^1, x^2,
  \ldots, x^{2\ed(x, y)}$ such that $x^0 = x$, $x^{2\ed(x, y)} = y$,
  and each $x^{i + 1}$ is either equal to $x^i$ or is obtained from
  $x^i$ by either a single insertion or a single deletion.\footnote{In
    order to ensure that each successive string is still a permutation
    of length at most $n$, it is useful to perform all of the
    deletions prior to the insertions.}

  By the no-collisions assumption (Lemma \ref{lemnocollisions}), we
  can select $b$ to be a sufficiently large constant such that $h$ is
  injective on all of the letters in all of the $x_i$'s with
  probability at least $1 - \frac{1}{4n p(n)}$. Therefore, by Lemma
  \ref{lemphiaftersingleinsertionperm}, for any $x^i, x^{i + 1}$, we
  have
  \begin{equation*}
    \Pr[\ed(\phi(x^{i + 1}), \phi(x^i)) > m] \leq
    \left(\frac{1}{2}\right)^{\Omega(m)} + \frac{1}{4n p(n)},
        \label{eqlogbalancing}
  \end{equation*}
  where the constant in the $\Omega(m)$ is independent of $p(n)$.

  If we apply Lemma \ref{lemgeometricvars2}, then we see that
  \begin{equation*}
    \Pr \left[\sum_{i = 0}^{2\ed(x, y) - 1}\ed(x^i, x^{i + 1})  > 2m \cdot \ed(x, y) \right]  \le \frac{1}{2^{\Omega(m)}} + \frac{4 \ed(x, y)}{4n \cdot p(n)} \le  \frac{1}{2^{\Omega(m)}} + \frac{1}{p(n)}.
  \end{equation*}
  By the triangle inequality, it follows that
  $$\Pr[\ed(\phi(x),\phi(y)) > m\ed(x, y)] \leq
  \left(\frac{1}{2}\right)^{\Omega(m)} + \frac{1}{p(n)}.$$
\end{proof}

It remains only to bound the run time necessary to compute $\phi(w)$
for $w \in \PP_{\le n}$.
\begin{lem}
  Let $w \in \PP_{\le n}$. Then we can compute $\phi(w)$ in time
  $O(n)$.
  \label{lempermdimredruntime}
\end{lem}
\begin{proof}
  This is straightforward except for one subtlety. In order to
  identify markers, we wish to identify for each window of size $2c +
  1$ which letter in that window has minimum hash (and whether that
  minimum is unique). This can be accomplished in linear time (and
  practically) using the Minimum on a Sliding Window Algorithm
  \cite{SlidingWindow}.
\end{proof}

We now present the full proof of Theorem \ref{thmdimredforperms}.

\begin{proof}[Proof of Theorem \ref{thmdimredforperms}]
  Lemma \ref{lempermdimredruntime} bounds the run time to $O(n)$. By
  Lemma \ref{lemshrinkperm}, the lengths of $\phi(x)$ and $\phi(y)$
  are at most $O(|x|/c)$ and $O(|y|/c)$. It remains to analyze the
  distortion of $\phi$. By Lemma \ref{lem:distortion-lower}, $\ed(x,
  y) \le c \cdot \ed(\phi(x), \phi(y))$. In order to complete the
  proof, we therefore wish to upper bound $\ed(\phi(x), \phi(y))$ in
  expectation by $O(\ed(x, y))$. By Lemma
  \ref{lemboundexpansionperm}, $$\Pr[\ed(\phi(x),\phi(y)) > m\ed(x,
    y)] \leq \left(\frac{1}{2}\right)^{\Omega(m)} + \frac{1}{p(n)},$$
  for all $m \in \mathbb{N}$. It follows that
  \begin{align*}
    \E[\ed(\phi(x), \phi(y)] & = \sum_{i = 0}^\infty \Pr[\ed(\phi(x), \phi(y)) > i] \\
    & \le \ed(x, y) \left(1 + \sum_{m = 1}^{\infty} \Pr[\ed(\phi(x), \phi(y)) > m \ed(x, y)]\right) \\
      & \le \ed(x, y) \left(1 +  \sum_{m = 1}^{n} \left(\left(\frac{1}{2}\right)^{\Omega(m)} + \frac{1}{p(n)}\right)\right) \qquad\text{[Using }\ed(\phi(x), \phi(y)) \le n\text{]} \\
      & = \ed(x, y) \left(O(1) + \frac{n}{p(n)}\right).
  \end{align*}
  For $p(n)$ sufficiently large, this is $O(\ed(x, y))$, as desired.

\end{proof}

\subsection{Dimension-Reduction Preliminaries}\label{secdimredprelims}

Before presenting our dimension-reduction map $\phi$ in the general
case, we provide preliminary definitions and lemmas.

\begin{defn}
Let $w \in \Sigma^{\le n}$ be a string. A contiguous substring $a$ in
$w$ is a \emph{maximally periodic substring} (with parameter $c$) if
the length of $a$ is at least $8c$, it is periodic with minimum period
at most $c$, and if extending $a$ one letter in either the left or the
right prohibits $a$ from having a minimum period at most $c$.
\end{defn}

Note that in a maximally periodic substring, the repeated substring need not appear an integer number of times.

\begin{lem}\label{lem:max-periodic-overlap}
Suppose $a$ and $b$ are maximally periodic substrings (with
parameter $c$) of $w$. If $a \neq b$, then $a$ and $b$ overlap in
fewer than $2c$ letters.
\end{lem}
\begin{proof}
Suppose $a$ and $b$ overlap in at least $2c$ letters. Let $p$ be the
minimum period of $a$ and $q$ be the minimum period of $b$. Because
$a$ has period $p$, $b$ has period $q$, and the overlap between $a$
and $b$ is of size at least $2c \ge p + q$, Lemma~\ref{lemgcdperiodic}
tells us that the overlap is $\gcd(p, q)$-periodic. Since any
$p$ consecutive letters in $a$ completely determine the values of the
other letters in $a$, and since the overlap of $a$ and $b$ consists of
at least $p$ consecutive letters, it follows that $a$ is
$\gcd(p, q)$-periodic. Similarly, $b$ is $\gcd(p,
q)$-periodic. Thus $p = q$.

Since $a$ and $b$ both have minimum period $p$ and have overlap length
at least $2c$, it follows that the string union of $a$ and $b$ has
minimum period $p$ as well. This contradicts the statement that $a$
and $b$ are maximally periodic substrings of $w$.
\end{proof}

\begin{defn}
Let $w \in \Sigma^{\le n}$ be a string. A contiguous substring $a$ in
$w$ is a maximally non-periodic substring (with parameter $c$) if $a$
is a maximal substring of $w$ that does not intersect with any
maximally periodic substrings.
\end{defn}

\begin{lem}\label{lem:local-permutation}
Let $w \in \Sigma^{\le n}$ and let $a$ be a maximally non-periodic substring
(with parameter $c$). Then, any $c+1$ subsequent overlapping blocks
(i.e., blocks starting at $c + 1$ consecutive letters) of at least
$8c$ characters in $a$ are distinct.
\end{lem}
\begin{proof}
Assume toward contradiction that there are two identical blocks $x,y$
of length at least $8c$ that start at indices $i,j$ (respectively)
such that $i<j$ and $j \leq i+c$. Then since $x_k = y_k$ for all $k
\le |x|$, and since $y_k = x_{k + (j - i)}$ for all $k \le |x| + i -
j$, it follows that $x_k = x_{k + (j - i)}$ for all $k \le |x| + i -
j$. Since $j - i \le c$, this means that $x$ is a $8c$-letter
substring with period at most $c$, contradicting the fact that $a$ is
a maximally non-periodic substring of $w$.
\end{proof}

\subsection{The Embedding in the General Case}\label{subsecdimred}
The embedding in the general case first has to identify all the
maximally periodic and non-periodic substrings. Then, it has different
rules as to where to place markers in each of the substrings. For
maximally non-periodic substrings, looking at overlapping $8c$-letter
substrings will result in something which is locally a permutation,
and will allow us to place markers similarly to the previous case (the
embedding for permutations). For maximally periodic substrings, the
period of the substring will determine the markers. The description of
the embedding appears as Algorithm~\ref{alg:embedding}.

\begin{algorithm}[tb]\caption{Dimension Reduction for Strings with Edit Distance\label{alg:embedding}}
Input: a string $w \in \Sigma^{\le n}$, a positive even integer $c$,
and a constant parameter $b \in \mathbb{N}$.

\begin{enumerate}
\item Select $h$ at random from a family $\mathcal{H}$ of $n$-wise
  independent hash functions mapping $\Sigma^{8c}$ to $\{0, 1\}^{b
    \log n}$.
\item Find all maximally periodic substrings (with parameter $c$) in
  $w$.
\item For each maximally periodic substring $a$ that was found in the
  previous step:
    \begin{enumerate}
    \item Let $p$ be the minimum period of $a$ and $s$ be the smallest multiple of $p$ satisfying $s \geq c$.
    \item Let $a_t$ be the first letter in $a$ such that $a_{t}
      a_{t+1} \cdots a_{t+p-1}$ is the lexicographically smallest
      cyclic shift of the $p$-letter substring that $a$ repeats. (Note
      that it appears at least once in $a$ since $\left|a\right| \geq
      8c$ and $p\le c$.)
    \item Place markers at each of the letters $a_t,a_{t+s},a_{t +
      2s},\ldots$ within the maximally periodic substring
      $a$. Additionally, place a marker after the final letter of $a$.
    \end{enumerate}
\item For each maximally non-periodic substring $x=x_1 \cdots x_d$:
    \begin{enumerate}
    \item Define a word $y$ of length $d-8c+1$ such that $y_i=x_i
      x_{i+1} \cdots x_{i+8c-1}$. Notice that the letters of $y$ are
      substrings of $x$.
    \item For each $i = \frac{c}{2}+1,\ldots,d-8c+1 - \frac{c}{2}$,
      place a marker at $x_i$ if, $h(y_i) < h(y_j)$ for all $j \neq i$
      satisfying $i - \frac{c}{2} \le j \le i + \frac{c}{2}$. That is,
      place a marker at $x_i$ if $h(y_i)$ is the unique minimum
      out of $h(y_{i - \frac{c}{2}}), \ldots, h(y_{i + \frac{c}{2}})$.
    \end{enumerate}
\item Partition $w$ into blocks in the following way: Each marker
  indicates the beginning of a block. For each block of size at least
  $c$ in a maximally non-periodic substring, subdivide it into blocks
  of size exactly $c$, except for the last one which will be of size
  at most $c$.
\item Return a string over the alphabet $\Sigma^*$ (i.e., whose
  letters are strings in $\Sigma$) where each block from the preceding
  step is a letter.
\end{enumerate}
\end{algorithm}

Our algorithm assumes access to an $n$-wise independent family
$\mathcal{H}$ of hash functions mapping strings of length $8c$ to the
set $\{0, 1\}^{b \log n}$. In Section
\ref{secdimensionreductionruntime} we will discuss how to implement
the algorithm in time $\tilde{O}(n)$, including how to efficiently
construct a family $\mathcal{H}$ which behaves as $n$-wise independent
with high probability. (Because the family $\mathcal{H}$ behaves as
$n$-wise independent only with high probability, the final algorithm
differs slightly from Algorithm \ref{alg:embedding}, but maintains the
original algorithm's high-probability results.)

In order to prove Theorem \ref{thm:main-dim-reduction-embedding}, we
divide the analysis into three parts. First, we bound the size of the
output strings $\phi(x),\phi(y)$ (Subsection
\ref{secboundlength}). Second, we bound the distortion in
Subsection~\ref{secbounddistortion}. We show that
$\frac{1}{2c}\ed(x,y) \leq \ed(\phi(x),\phi(y))$, and then we show
that the probability of $\ed(\phi(x), \phi(y))$ being greater than
$\ed(x, y)$ by a factor of $m$ decreases exponentially in $m$; this is
the most difficult step in the analysis. Finally, we demonstrate how
to construct $\phi(x)$ in time $O(n \log c)$ (Subsection
\ref{secdimensionreductionruntime}). Combining these, we are able to
prove Theorem \ref{thm:main-dim-reduction-embedding} in Subsection
\ref{secproofofthm}.

\subsubsection{Bounding the Length of the Output Strings}\label{secboundlength}

\begin{lem}
Let $x \in \Sigma^{\le n}$. When computing $\phi(x)$, at most $5$
markers are placed in any given $c/2$ consecutive letters of $x$.
\end{lem}
\begin{proof}
Consider some $c/2$ consecutive letters $a = a_1 \cdots a_{c/2}$ of
$x$. The only way for two markers placed by a maximally periodic
substring to be within the same $c/2$ positions is for one of those
markers to be after the end of the periodic substring. Thus any
maximally periodic substring intersecting $a$ contributes at most two
markers. By Lemma \ref{lem:max-periodic-overlap}, the letters in $a$
can appear in at most two maximally periodic substrings, which
therefore contribute at most four markers.

On the other hand, we claim that at most one marker can be placed in a
given $c / 2$ letters by maximally non-periodic substrings. Indeed, if
markers $w_i$ and $w_j$ are placed in the same $c/2$ letters, then we
get that for the corresponding $y_i,y_j$, we have $h(y_i) < h(y_j)$
and $h(y_j) < h(y_i)$, a contradiction.
\end{proof}

\begin{cor}
Let $w \in \Sigma^{\le n}$. When computing $\phi(w)$, the total number
of blocks that intersect any $mc$-letter substring is at most
$O(m)$. Hence $|\phi(w)| \le O(n / c)$.
  \label{corboundlength}
\end{cor}
\begin{proof}
By the preceding lemma, every chunk of $c/2$ letters contains at most
$5$ markers, bounding the number of markers to $10m \in O(m)$.

After markers are placed, blocks between markers in non-periodic
substrings may be further partitioned. This further partitioning can
be seen as adding a new sub-marker every $c$ letters until the next
marker. Note that this process never places two sub-markers in the
same $c/2$ letters, bounding the number of submarkers in our
$mc$-letter substring to $2m$.

Therefore, the total number of blocks that intersect the $mc$-letter
substring is $O(m)$.
\end{proof}

\subsubsection{Bounding the Distortion}\label{secbounddistortion}

\begin{observation}
Each block is of size at most $c$ within maximally non-periodic
substrings, and is size at most $2c - 1$ within maximally periodic
substrings. Blocks that start in a maximally non-periodic substring
and end in a maximally periodic substring are also of size at most $2c
- 1$.
\end{observation}

\begin{lem}
  For any $x,y$, $\ed(x,y) \leq 2c \ed(\phi(x),\phi(y))$.
  \label{lemboundcontraction}
\end{lem}
\begin{proof}
The proof follows exactly as that of Lemma~\ref{lem:distortion-lower},
except with $c$ replaced with $2c$.
\end{proof}

\begin{lem}\label{clm:minimal-letter}
Consider $W \in \mathbb{N}$. Let $x$ be a string of length $m \cdot
2W$ such that any contiguous substring of length at most $2W$ is a
permutation. Consider an $|x|$-wise independent hash function
$h$, and condition on $h$ having no collisions on the letters of
$x$. For $W < i \leq |x| - W$, we say that $x_i$ is a minimal letter
if $h(x_i) = \min_{i-W \leq j \leq i+W}{h(x_j)}$. Then,
\[
\Pr[x\text{ does not contain a minimal letter}] \leq \left(\frac{1}{2}\right)^{m-2}.
\]
\end{lem}
\begin{proof}
  Break $x$ into subwords $x^1, x^2, \ldots, x^{2m}$ of size $W$. Let
  $h_{\min}\left(x^i\right)$ denote the minimum hash of the letters in
  $x^i$. Because each window of $2W$ letters is a permutation in $x$,
  $x^i$ and $x^{i + 1}$ are supported by disjoint sets of letters for
  all $i$. Therefore $h_{\min}\left(x^i\right) \neq h_{\min}\left(x^{i
    + 1}\right)$ for all $i$.

  We begin by analyzing the structure of $x$ under the condition that
  $x$ contains no minimal letters. Suppose there is some $s$ such that
  $h_{\min}\left(x^{s - 1}\right) > h_{\min}\left(x^s\right) $, and pick a minimal such
  $s$. Then in order so that $x^s$ not contain a minimal letter, it
  must be that either $s = 2m$ or that $h_{\min}\left(x^{s}\right) > h_{\min}\left(x^{s
    + 1}\right)$. Continuing like this, we see that
  \begin{equation}
    h_{\min}\left(x^{s - 1}\right) > h_{\min}\left(x^{s}\right) > h_{\min}\left(x^{s + 1}\right) > \cdots > h_{\min}\left(x^{2m}\right).
    \label{eqchaindown}
  \end{equation}
  Moreover, since $s$ is the minimal $s$ satisfying $h_{\min}\left(x^{s -
    1}\right) > h_{\min}\left(x^s\right)$, we must have
  \begin{equation}
    h_{\min}\left(x^{1}\right) < h_{\min}\left(x^{2}\right) < h_{\min}\left(x^{3}\right) < \cdots < h_{\min}\left(x^{s - 1}\right).
    \label{eqchainup}
  \end{equation}
  Notice that if there does not exist an $s$ for which $h_{\min}\left(x^{s -
    1}\right) > h_{\min}\left(x^s\right)$, then \eqref{eqchainup} holds for $s = 2m +
  1$. Therefore, as long as $x$ contains no minimal letters, then
  \eqref{eqchaindown} and \eqref{eqchainup} hold for some $s \in \{2,
  \ldots, 2m + 1\}$.

  To complete the proof, we will show that, given an arbitrary $x$,
  the probability of \eqref{eqchaindown} and \eqref{eqchainup} holding
  for some $s$ is at most $1/2^{m - 2}$. If we consider the cases of
  $s \le m + 1$ and $s > m + 1$ separately, then it suffices to prove in each
  case that the probability of \eqref{eqchaindown} and
  \eqref{eqchainup} holding is at most $1/2^{m - 1}$. Since the
  problem is symmetrically defined for these two cases, we will
  consider only $s \le m + 1$. In particular, we will show that for a
  given $x$,
  \begin{equation}
    \Pr\Big[h_{\min}\left(x^{m + 1}\right) > h_{\min}\left(x^{m + 2}\right) > h_{\min}\left(x^{m + 3}\right) > \cdots > h_{\min}\left(x^{2m}\right)\Big] \le \left(\frac{1}{2}\right)^{m - 1}.
    \label{eqsimplifiedchain}
  \end{equation}

  Define $a_i$ to be the number of letters appearing in $x^{m + i}$
  but not in any of $x^{m + 1}, \ldots, x^{m + i - 1}$. Define $b_i
  = \sum_{j = 1}^{i}a_i$ to be the number of distinct letters
  appearing in $x^{m + 1}, x^{m + 2} \ldots, x^{m +i}$. Suppose that
  \begin{equation}
    h_{\min}\left(x^{m + 1}\right) > h_{\min}\left(x^{m + 2}\right) > h_{\min}\left(x^{m + 3}\right) > \cdots > h_{\min}\left(x^{m + i - 1}\right).
    \label{eqinductivechain}
  \end{equation}
  Then in order for $h_{\min}\left(x^{m + i}\right)$ to be smaller
  than $h_{\min}\left(x^{m + i - 1}\right)$, some letter $l$ in $x^{m
    + i}$ not appearing in any earlier $x^{m + j}$ must have the
  smallest hash out of the letters in $x^{m + 1}x^{m + 2}\cdots x^{m +
    i}$. For a given $l \in x^{m + i}$ not appearing in any earlier
  $x^{m + j}$, the probability of $l$'s hash being smallest is
  $1/b_i$. By the union bound, it follows that the probability of
  $h_{\min}\left(x^{m + i}\right) < h_{\min}\left(x^{m + i -
    1}\right)$ given \eqref{eqinductivechain} is at most $a_i/b_i$.

  Hence the probability of \eqref{eqsimplifiedchain} occurring
  for a given $x$ is no greater than
  \begin{equation*}
      \prod_{i = 2}^{m} \frac{a_i}{b_i} = \prod_{i = 2}^{m} \frac{a_i}{\sum_{j = 1}^i a_i}.
  \end{equation*}
  Because each window of $2W$ letters in $x$ is a permutation, $a_1$
  and $a_{2}$ are both $W$. This gives us
  \begin{equation*}
      \prod_{i = 2}^{m} \frac{a_i}{2W + \sum_{j = 3}^i a_i}  \le \prod_{i = 2}^{m} \frac{a_i}{2W} \le \left(\frac{1}{2}\right)^{m - 1},
  \end{equation*}
  as desired.
\end{proof}

The above claim allows us to bound the probability of our algorithm
failing to place any markers in a long substring of consecutive
characters.

\begin{lem}\label{lem:consecutive-marker}
Consider $w \in \Sigma^{\le n}$. Consider a $(m + 10) \cdot c$ letter
substring $u$ of $w$, and condition on $h$ being injective on the
$8c$-letter contiguous substrings of $u$. (Here, $h$ is as defined in
Algorithm \ref{alg:embedding}.)  Then the algorithm places a marker in
$u$ with probability at least $1-\left(\frac{1}{2}\right)^{m-2}$.
\end{lem}
\begin{proof}
  Consider a substring $u$ of $(m + 10) \cdot c$ consecutive
  letters. If $u$ intersects a maximally periodic substring by at
  least $2c$ letters, or contains the letter after a maximally
  periodic substring, then $u$ must contain a marker from that
  maximally periodic substring. Thus in order for $u$ to be at risk of
  not containing a marker, all but the final $2c - 1$ letters of $u$
  must be elements of the same maximally non-periodic substring.

Thus we need only consider the case where $u$ contains at least $(m +
8) \cdot c$ consecutive letters all in the same maximally non-periodic
substring. Let $v$ be the string of $y$-values corresponding to the
first $mc$ of those $(m + 8)c$ consecutive letters. Because the
letters are in a maximally non-periodic substring, each consecutive
$c$ $y$-values in $v$ are distinct (by
Lemma~\ref{lem:local-permutation}). By the assumption that $h$ is
injective on the $8c$-character substrings of $u$, we have that $h_k$
is injective on the letters of $v$.  By
Lemma~\ref{clm:minimal-letter}, with probability at least $1 -
\left(\frac{1}{2}\right)^{m-2}$, there is a minimal $y$ value in $v$,
which leads to a marker being placed by the algorithm at the
corresponding letter in $u$.
\end{proof}

\begin{lem}
Let $m \in \mathbb{N}$. Let $x$ be a string and $y$ be a string that
is derived from $x$ by a single insertion before letter
$x_i$. Condition on the fact that $h$ is injective on the $8c$-letter
substrings of $x$ and $y$.

Then, $\Pr[\ed(\phi(x),\phi(y)) > m] \leq
\left(\frac{1}{2}\right)^{\Omega(m)}$.
\label{lemphiaftersingleinsertion}
\end{lem}
\begin{proof}
Let $i$ be the index such that $y$ is generated from $x$ by adding a
character at index $i$. That is, $x_1 \cdots x_{i-1} = y_1 \cdots
y_{i-1}$ and $x_i x_{i+1} \cdots = y_{i+1} y_{i+2} \cdots$.

We claim that for any $j < i - 16c - c/2$, there will be a marker at
$x_j$ if and only if there is a marker at $y_j$. In particular, $x_j$
belongs to a maximally periodic substring if and only if $y_j$ belongs
to the same maximally periodic substring, except that the periodic
substring containing $y_j$ may be truncated or extended by the
insertion of $y_i$.  (Note that the insertion occurs more than $8c$
characters after position $j$, and thus cannot get rid of the
maximally periodic substring entirely). Moreover, for the same reason,
$x_j$ appears immediately after a maximally periodic substring if and
only if $y_j$ appears immediately after the same maximally periodic
substring. Therefore, for $j < i - 16c - c/2$, $x_j$ will be a marker
due a maximally periodic substring if and only if $y_j$ is a marker
due to a maximally periodic substring. On the other hand, if both
$x_j$ and $y_j$ belong to maximally non-periodic substrings, then the
non-periodic substring containing $x_j$ will extend to position $x_{j
  + 8c + c/2}$ if and only if the non-periodic substring containing
$y_j$ extends to position $y_{j + 8c + c/2}$. (Here, we are using that
the insertion occurs more than $8c$ characters after position $j + 8c
+ c/2$ and thus cannot affect for letters in positions $\le j + 8c +
c/2$ whether they are contained in a maximally non-periodic
substring). Therefore, for $j < i - 16c - c/2$, given that $x_j$ and
$y_j$ are contained in maximally non-repetitive substrings, one will
be a marker if and only if the other is as well.

Now, we claim that the blocks that contain only characters before
$x_{i-16c-c/2 - 1}$ or $y_{i-16c-c/2 - 1}$ are identical and appear in
both $\phi(x)$ and $\phi(y)$. This is because, in general, any edit
which modifies only markers more than one character to the right of a
block will not affect that block.

In order to examine which blocks after the insertion are affected by
the insertion, we consider two cases.

\textbf{Case 1:} $x_i$ does not belong to a maximally periodic
substring that extends until at least $x_{i+8c-1}$. In turn, this
means that $y_{i + 1}$ does not belong to a maximally periodic
substring that extends until at least $y_{i + 8c}$, since otherwise,
that periodic substring would still be present in $x$ beginning at
$x_i$. Therefore, the maximally periodic substrings in $x$ which
intersect $x_{i + 8c - 1}x_{i + 8c}\cdots$ are identical to those in
$y$ which intersect $y_{i + 8c}y_{i + 8c + 1}\cdots$. It follows that
any for $j \ge i + 8c$, $x_j$ will be a marker due a maximally
periodic substring if and only if $y_{j + 1}$ is a marker due to a
maximally periodic substring. Moreover, it follows that for $j \ge i +
8c$, $x_j$ will be a member of a maximally non-periodic substring if
and only if $y_{j + 1}$ is a member of a maximally non-periodic
substring. Therefore, for $j \ge i + 8c + c/2$, $x_j$ will be a marker
in a maximally non-periodic substring if and only if $y_{j + 1}$ is a
marker in a maximally non-periodic substring.

So far we have established that for $j \ge i + 8c + c/2$, $x_j$ is a
marker if and only if $y_{j + 1}$ is a marker. Recalling by assumption
that $h$ is injective on the $8c$-letter substrings of $x$ and $y$, by
Lemma~\ref{lem:consecutive-marker}, with probability at least
$1-\left(\frac{1}{2}\right)^{m-2}$, there will be a marker in the $(m
+ 10) \cdot c$ characters after $x_{i+8c + \frac{c}{2}}$ (if such
characters exist). Assume that this marker is located at $x_j$. Then a
marker will also be placed at $y_{j+1}$, and since all of the markers
to the right of $x_j$ are identical to those to the right of $y_{j +
  1}$, all the blocks that contain the characters starting from $x_j$
or $y_{j+1}$ will be identical in $\phi(x)$ and $\phi(y)$. We get that with
probability at least $1-\left(\frac{1}{2}\right)^{m-2}$, only the
blocks that contain the characters from $x_{i- 16c - c/2 - 1}, \ldots,
x_{i+(m + 18 + 1/2)c}$ or $y_{i - 16c - c/2 - 1}, \ldots ,y_{i+(m + 18 +
  1/2)c+1}$ can differ between $\phi(x)$ and $\phi(y)$. By Corollary
\ref{corboundlength}, there are at most $O(m)$ blocks that need to be
edited in order to transform $\phi(x)$ to $\phi(y)$.

\textbf{Case 2:} $x_i$ belongs to a maximally periodic substring that
extends until at least $x_{i+8c-1}$. By Lemma
\ref{lem:max-periodic-overlap}, there is only a single maximally
periodic substring $u$ containing both $x_i$ and $x_{i + 8c -
  1}$. Because $u$ contains at least $8c$ letters to the right of the
insertion, those $8c$ characters will still be part of a maximally
periodic substring in $y$, and the endpoint of the maximally periodic
substring will be unchanged. Consequently, if we denote by $x_{i^*}$
the last character of $u$, then $y_{i^* + 1}$ will be the last character
in $y$ which is contained in a maximally periodic substring also
containing $y_{i + 1}$. As a result, the maximally periodic substrings
and the maximally non-periodic substrings in $x$ which intersect the
letters to the right of $x_{i^*}$ are identical to those in $y$ which
intersect the letters to the right of $y_{i^* + 1}$. Thus the markers
to the right of $x_{i^*}$ will be identical to those to the right of
$y_{i^* + 1}$. Moreover, there will be a marker at each of $x_{i^* + 1}$
and $y_{i^* + 2}$ because they each marks the end of a maximally
periodic substring. Therefore, the blocks beginning at $x_{i^* + 1}$ will
be identical to those beginning at $y_{i^* + 2}$.

Now we consider letters between $x_i$ and $x_{i^*}$ and the letters
between $y_i$ and $y_{i^*}$ (inclusive). Note that $x_i \cdots x_{i^*}
= y_{i+1} \cdots y_{i^*+1}$, and both are suffixes of maximally
periodic substrings with the same periods as each other. Recall that
$u$ is the maximally periodic substring of which $x_i \cdots x_{i^*}$
is a suffix; and let $v$ denote the maximally periodic substring of
which $y_{i + 1} \cdots y_{i^* + 1}$ is a suffix. By Lemma
\ref{lem:max-periodic-overlap}, the only markers in the substrings
$x_{i + 2c} \cdots x_{i^* - 2c}$ and $y_{i + 2c + 1} \cdots y_{i^* -
  2c + 1}$ are due to $u$ and $v$, respectively (and not some other
periodic substring overlapping $u$ or $v$). Denote by $j_1$ the index
of the first marker placed in $x_{i + 2c} \cdots x_{i^* - 2c}$, and by
$j_1'$ the index of the first marker placed in $y_{i+2c+1} \cdots
y_{i^*- 2c + 1}$. Additionally, denote by $j_2$ the index of the last
marker placed before or at $x_{i^* - 2c}$ and by $j_2'$ the index of
the last marker placed before or at $y_{i^* - 2c +1}$.

Each of the blocks in $x$ between $x_{j_1}$ and $x_{j_2}$ and each
of the blocks in $y$ between $y_{j_1'}$ and $y_{j_2'}$ are
identical.\footnote{In particular, if $p$ is the period of the
  maximally periodic substrings containing the blocks, then the length
  of each block is the smallest multiple $s$ of $p$ satisfying $s \ge
  c$; and the contents of each block is oriented to begin on the
  lexicographically smallest cyclic shift of the $p$-letter substring
  being repeated.} With this in mind, we propose the following
sequence of edits from $\phi(x)$ to $\phi(y)$:
\begin{enumerate}
\item The blocks that contain only letters before $x_{i-16c - c/2 - 1}$ or
  $y_{i-16c-c/2 - 1}$ are identical, and therefore incur no edits.
\item There are $O(c)$ characters in the subwords $x_{i- 16c - c/2 - 1}
  \cdots x_{j_1 - 1}$ and $y_{i- 16c - c/2 - 1} \cdots y_{j_1' - 1}$. By
  Corollary \ref{corboundlength}, the blocks containing characters
  from the subword $x_{i- 16c - c/2 - 1} \cdots x_{j_1 - 1}$ can
  be transformed into the blocks containing characters from the
  subword $y_{i- 16c - c/2 - 1} \cdots y_{j_1' - 1}$ in $O(1)$ operations.
\item The blocks that correspond to $x_{j_1} \cdots x_{j_2 - 1}$ and
  $y_{j_1'} \cdots y_{j_2' - 1}$ are identical, except that one of these
  substrings may include $O(1)$ additional blocks (since the blocks in
  the maximally periodic substring $u$ may appear at different offsets
  than the blocks in $v$). Thus we can map the blocks in these two
  substrings to each other in $O(1)$ operations.
\item There are $O(c)$ characters in the subwords $x_{j_2} \cdots
  x_{i^*}$ and $y_{j_2'} \cdots y_{i^* + 1}$. By Corollary
  \ref{corboundlength}, the blocks containing characters from the
  subword $x_{j_2} \cdots x_{i^*}$ can therefore be transformed
  into the blocks containing characters from the subword $y_{j_2'}
  \cdots y_{i^* + 1}$ in $O(1)$ operations.
\item The blocks beginning at $x_{i^* + 1}$ are identical to those
  beginning at $y_{i^* + 2}$, requiring no additional edits.
\end{enumerate}
The total number of edits needed to transform $\phi(x)$ to $\phi(y)$ in this case is $O(1)$.

In conclusion, in either case, the total number of edits is $O(m)$ (with probability at least
$1-\left(\frac{1}{2}\right)^{m-2}$), as desired.
\end{proof}

We are now prepared to bound the distortion of our dimension-reduction
map $\phi$.
\begin{lem}
  Let $x$ and $y$ be strings. Then for $m \in
  \mathbb{N}$,
  $$\Pr[\ed(\phi(x),\phi(y)) > m\ed(x, y)] \leq
  \left(\frac{1}{2}\right)^{\Omega(m)} + \frac{1}{\poly(n)},$$ where
  $\poly(n)$ is an arbitrarily large polynomial in $n$ of our choice
  controlled by the constant $b$.
  \label{lemboundexpansion}
\end{lem}
\begin{proof}
  The proof follows exactly as that of Lemma
  \ref{lemboundexpansionperm} for permutations. The only difference is
  that here Lemma \ref{lemphiaftersingleinsertion} should be used in
  place of \ref{lemphiaftersingleinsertionperm}.
\end{proof}

\subsubsection[Bounding the Run Time by O(n log c)]{Bounding the Run Time by $O(n \log c)$}\label{secdimensionreductionruntime}

In this section, we discuss how to implement our dimension reduction
algorithm in time $O(n \log c)$.

We first discuss how to identify maximally periodic substrings.  The
following lemma allows us to quickly detect whether a $4c$-letter
string is periodic with a small period.
\begin{lem}
  Let $a = a_1 \cdots a_{4c}$ be a string of length $4c$. We can
  detect whether $a$ is periodic with a period $\le c$ in time
  $O(|a|)$. Additionally, if $a$ is periodic, we return its
  (minimum) period $p$.
  \label{lemdetectperiodic}
\end{lem}
\begin{proof}
  Consider the suffix array $A$ of $a$. The $i$-th entry of $A$ is the
  $i$-th lexicographically smallest suffix of $a$. It is well known
  that $A$ can be constructed in time $O(|a|)$
  \cite{farach1997optimal}.

  Suppose $a$ is periodic with minimum period $p \le c$. Then
  $$a_1 \cdots a_{4c - p} = a_{p + 1} \cdots a_{4c}.$$ Moreover, we
  claim that, out of all the proper suffixes of $a$, $a_{p + 1} \cdots
  a_{4c}$ has the longest common prefix with $a$. Indeed suppose there
  is some $1 \le q < p$ such that $a_1 \cdots a_{4c - p} = a_{q + 1} \cdots
  a_{4c - p + q }$. Then the first $4c - p$ letters of $a$ are both
  $p$-periodic and $q$-periodic; by Lemma \ref{lemgcdperiodic}, since
  $4c - p \ge p + q$, it follows that the first $4c - p$ letters of
  $a$ are $\gcd(p, q)$-periodic. Since all of $a$ is $p$-periodic, it
  follows that all of $a$ is $\gcd(p, q)$-periodic, contradicting that
  $p$ is the minimum period of $a$. Thus no such $q$ exists, meaning
  that out of all the proper suffixes of $a$, $a_{p + 1} \cdots
  a_{4c}$ has the longest common prefix with $a$. It follows that
  $a_{p + 1} \cdots a_{4c}$ will appear adjacently to $a_1a_2 \cdots
  a_{4c}$ in the suffix array $A$ (and will, in fact, appear right
  before $a_1a_2 \cdots a_{4c}$).

  Thus, in the event that $a$ is periodic with minimum period $p \le
  c$, we can determine $p$ by examining the element in $A$ preceding
  $a_1a_2 \cdots a_{4c}$. Therefore, the following algorithm can be
  used to determine whether $a$ has minimum period $\le c$: Define $p
  + 1$ to be the index in $a$ of the first letter of the element of
  $A$ preceding $a_1a_2 \cdots a_{4c}$; check whether $p \le c$ and
  whether $a$ is $p$-periodic; if $a$ is, return $p$, and otherwise
  return that $a$ is not periodic with period $\le c$.
\end{proof}

Armed with the preceding lemma, we can identify the maximally
periodic substrings of $w$ in time $O(n)$.

\begin{prop}
  The maximally periodic substrings of $w$ can be identified in time $O(n)$.
  \label{lemfastmaximallyperiodicdetect}
\end{prop}
\begin{proof}
  We identify the maximally periodic substrings as follows. Break $w$
  into disjoint chunks $u_1, \ldots, u_m$ of size $4c$ (ignoring up to
  $4c - 1$ letters at the end). Because maximally periodic substrings
  are length $\ge 8c$, each maximally periodic substring must entirely
  contain some chunk $u_i$. Therefore, we can detect the maximally
  periodic substrings as follows. For $i$ from $1$ to $m$, determine
  whether $u_i$ is periodic with period $\le c$, and if it is let $p$
  be its period; we can do this in time $O(|u_i|) = O(c)$ per $u_i$ by
  Lemma \ref{lemdetectperiodic}. If $u_i$ is part of a larger
  maximally periodic substring, then that substring must also have
  period $p$. Thus we can detect whether $u_i$ is part of a maximally
  periodic string by extending $u_i$ to the left and right as far as
  we can go while still maintaining period $p$ (This is easy to do in
  time proportional to the length of the resulting extended $u_i$). If
  the result is length $\ge 8c$, then we have identified a maximally
  periodic substring. If we identify a maximally periodic substring,
  then we move on to the first $u_j$ with $j > i$ such that $u_j$ is
  not contained in the maximally periodic substring; otherwise, we
  move on to $u_{i + 1}$ and repeat the process.

For each of the $\lfloor n / (4c) \rfloor$ $u_i$'s we spend $O(c)$
work finding the period $p$ of $u_i$ (if there is one). We then either
spend time $O(c)$ extending $u_i$ to a periodic string of length $<
8c$, or we extend $u_i$ to a maximally periodic substring $x$ in time
$O(|x|)$. The total time spent on $u_i$'s which don't yield maximally
periodic substrings is $O(\frac{n}{c} \cdot c) = O(n)$. Since we spend
$O(|x|)$ time successfully finding each maximally periodic substring
$x$, the total time spent extending $u_i$'s to maximally periodic
substrings is $O(n)$. Therefore, the total time spent to detect the
maximally periodic substrings is $O(n)$.
\end{proof}

For each maximally periodic substring $x$ with period $p$, we need to
determine the offset $i \in [p]$ which minimizes $x_ix_{i+1}\cdots
x_{i + p - 1}$ lexicographically.\footnote{Note that this offset is
  unique. In particular, if $x_i \cdots x_{i + p - 1} = x_j \cdots
  x_{j + p - 1}$ for $i < j \le p$ then it follows that $x_i \cdots
  x_{j + p - 1}$ is $(j - i)$-periodic. By Lemma \ref{lemgcdperiodic},
  we get that $x_i \cdots x_{j + p - 1}$ is $\gcd(p, j -
  i)$-periodic. Thus $x$ is $\gcd(p, j - i)$-periodic, a
  contradiction.} This can be performed in time $O(c)$ by building the
suffix array $A$ for the substring $x_1x_{2} \cdots x_{2p - 1}$, and
then comparing the positions within the suffix array of $x_ix_{i + 1}
\cdots x_{2p - 1}$ for each $i \in [1, p]$. Thus we can complete the
third step of the algorithm for all of the maximally periodic
substrings in time $O(n)$.

So far we have explained how to determine the markers in maximally
periodic strings in time $O(n)$. In order to show that the entire
algorithm can be performed in time $O(n \log c)$, we must show how to
efficiently determine the markers within maximally non-periodic
strings. The main difficulty in this is computing the hashes of each
of the $8c$-letter substrings. In particular, once these are computed,
the markers can be determined using the Minimum on a Sliding Window
Algorithm as in the proof of Lemma \ref{lempermdimredruntime}.

Therefore, we devote the remainder of the section to constructing a
family $\mathcal{H}$ of hash functions from $\Sigma^{8c}$ to $\{0,
1\}^{b \log n}$ such that with high probability $\mathcal{H}$ behaves
as $n$-wise independent (that is, for any set $A$ of size $n$, with
high probability the family $\mathcal{H}$ is independent on $A$), and
such that $h \in \mathcal{H}$ can be evaluated on every $8c$-letter
substring of $w$ in total time $O(n \log c)$.

Recall from \cite{amazinghash} that in time $O(n)$ one can construct a
family $\mathcal{S}$ of hash functions from $\{0, 1\}^{\Theta(\log
  n)}$ to $\{0, 1\}^{b \log n}$ such that with high probability,
$\mathcal{S}$ behaves as $n$-wise independent, and such that $s \in
\mathcal{S}$ can be evaluated in constant time. It therefore suffices
to construct a family $\mathcal{G}$ of hash functions mapping
$\Sigma^{8c}$ to $\{0, 1\}^{\Theta(\log n)}$ so that with high
probability, $g \in \mathcal{G}$ is injective on the $8c$-letter
substrings of $w$. Given such a family $\mathcal{G}$, we can then
define $h = g \circ s$ in order to obtain a family $\mathcal{H}$ which
behaves as $n$-wise independent from $\Sigma^{8c}$ to $\{0, 1\}^{b
  \log n}$ with high probability.

For convenience, we will assume that $8c$ is a power of two, though
the argument can easily be adapted to the more general
setting.\footnote{This is done by first hashing the windows of size
  $2^{\lfloor \log 8c \rfloor}$, and then computing a hash of each
  window of size $8c$ by hashing together the hash for the overlapping
  windows of size $2^{\lfloor \log 8c \rfloor}$ which constitute the
  larger window of size $8c$.} Our construction relies on the idea of
a Merkle tree. An element $g$ of $\mathcal{G}$ is determined by $\log
8c + 1$ pairwise independent hash functions $f_0, f_1, f_2, \ldots,
f_{\log 8c}$, with each $f_i$ mapping $2 r \log n$ bits to $r \log n$
bits (where $r$ is a constant of our choice). Given these $f_i$
functions, we recursively define a function $g^*$ such that for a
string $a = a_1 \cdots a_{2^k}$ (with $2^k \le 8c$) we have that
$g^*(a) = f_k(g^*(a_1 \cdots a_{2^{k - 1}}), g^*(a_{2^{k - 1} + 1}
\cdots a_{2^{k}}))$ when $k \ge 1$ and by $g^*(a) = f_0(a_1)$ when $k
= 0$. Finally, we define our function $g$ by $g(a) = g^*(a)$ for
strings of length $8c$. The next lemma establishes that $g \in
\mathcal{G}$ is injective on the $8c$-letter substrings of $w$ with
high probability, as desired.
\begin{lem}
  Suppose $r$ is sufficiently large (remember each $f_i$ maps to $r
  \log n$ bits), and let $g$ be selected at random from the family
  $\mathcal{G}$ described above. Then $g$ is injective on the
  $8c$-letter substrings of $w$ with probability $1 -
  \frac{1}{\poly(n)}$.
\end{lem}
\begin{proof}
   Let $u, v \in \Sigma^{8c}$ be distinct. By Lemma
   \ref{lemnocollisions}, if we select $r$ large enough, then with
   probability at least $1 - \frac{1}{\poly(n)}$ for a polynomial of
   our choice, there will be no collisions in any of the applications
   of hash function $f_0$ when computing $g(u)$ and $g(v)$. The same
   holds for each of the other $f_i$'s. In particular, if there are no
   collisions when we apply $f_{i - 1}$, then $f_i$ will be fed distinct inputs
   and with high probability will also have no collisions (this is why we use a different hash function at each level). By the
   union bound, it follows that with probability $1-
   \frac{1}{\poly(n)}$, $g^*(u) \neq g^*(v)$.

   Applying the union bound again to all pairs $u, v$ of $8c$-letter
   substrings of $w$, it follows that with probability $1 -
   \frac{1}{\poly(n)}$, $g$ is injective on them.
\end{proof}

By defining $\mathcal{G}$ in this way, we are able to quickly compute
all of the $8c$-letter window hashes using dynamic programming.
\begin{thm}
  Let $g \in \mathcal{G}$ be selected at random. Suppose $a_1 \cdots
  a_m$ is a string over an alphabet $\Sigma$. Then we can compute $g(a_i \cdots a_{i + 8c - 1})$ for all $i$ in time $O(m \log c)$.
  \label{thmfastwindow}
\end{thm}
\begin{proof}
  This is a simple application of dynamic programming. In particular,
  we perform $O(\log c)$ iterations. In the $i$-th iteration we
  compute the $g^*$ hashes of all the substrings of length $2^{i - 1}$
  (allowing overlap). Using the results of the previous iteration,
  each iteration takes time $O(m)$. Thus the total run time is $O(m
  \log c)$.
\end{proof}

\subsubsection{Proof of Theorem \ref{thm:main-dim-reduction-embedding}}\label{secproofofthm}

We are now prepared to prove that $\phi$ is a dimension-reduction map
satisfying the properties claimed in Theorem
\ref{thm:main-dim-reduction-embedding}.

\begin{proof}
  Consider two strings $x, y \in \Sigma^{\le n}$. By Corollary
  \ref{corboundlength}, $\phi(x)$ and $\phi(y)$ are lengths $O(|x|/c)$
  and $O(|y|/c)$. By Lemma \ref{lemboundcontraction}, $\ed(x, y) \le
  2c \cdot \ed(\phi(x), \phi(y))$. By Lemma \ref{lemboundexpansion},
  for $m \in \mathbb{N}$, $$\Pr[\ed(\phi(x),\phi(y)) > m\ed(x, y)]
  \leq \left(\frac{1}{2}\right)^{\Omega(m)} + \frac{1}{\poly(n)},$$
  for all $m \in \mathbb{N}$ and for a polynomial $\poly(n)$ of our
  choice. Exactly as in the proof of Theorem \ref{thmdimredforperms},
  as long as we pick $\poly(n)$ sufficiently large, it follows that
  $\E[\ed(\phi(x), \phi(y))] \le O(\ed(x, y))$, and subsequently the
  expected distortion of $\phi$ is $O(c)$. Finally, the discussion in
  Section \ref{secdimensionreductionruntime} establishes that
  $\phi(x)$ and $\phi(y)$ can each be computed in time $O(n \log c)$,
  as desired.
\end{proof}

\subsection{An Approximation Algorithm}\label{secdimredapproxalg}

Prior to the approximation algorithms of
\cite{ApproxSubPolyDistortion} and \cite{ApproxPolyLog}, the state of
the art for approximating edit distance was the algorithm of
\cite{DimensionReduction}, which obtained distortion
$$\tilde{O}\left(\min\left(n^{\frac{1}{3} + \frac{2}{3 \log \log \log
    n}}, \ed(x, y)^{\frac{1}{2} + \frac{1}{\log \log \log n}} \right)
\right).$$

Moreover, prior to our results in Section \ref{secblackbox}, the
algorithm of \cite{DimensionReduction} remained the best approximation
algorithm for the approximate alignment problem.

The approximation algorithm of \cite{DimensionReduction} is built on
top of the dimension-reduction map presented in the same
paper. Consequently, the improvements in dimension-reduction
distortion presented in this section can be used to improve the
approximation algorithm. In particular, the algorithm can be updated
to achieve distortion
$${O}\left(\min\left(n^{\frac{1}{3}}, \ed(x, y)^{\frac{1}{2}}\right)
\right).$$

In order so that our dimension-reduction map $\phi$ can be used by the
algorithm of \cite{DimensionReduction}, one must prove a subtle
property of $\phi$. Because the property may be useful in future work,
we state it here. The proof can be found in Appendix
\ref{secproofoflemapprox}.

\begin{lem}
 Consider $x, y \in \Sigma^{\le n}$. Pick $m \in \mathbb{N}$. Then
 with probability $1- \left(\frac{1}{2}\right)^{\Omega(m)} -
 \frac{1}{\poly(n)}$ (for a polynomial of our choice), there is a
 sequence of at most $m \ed(x, y)$ edits from $\phi(x)$ (with
 parameter $c$) to $\phi(y)$ with the following property. If a block
 $u$ in $\phi(x)$ is not modified by the edits and is mapped to a
 block $v$ in $\phi(y)$, then the start position of block $u$ in $x$
 is within $O(\ed(x, y) + c)$ of the start position of block $v$ in
 $y$.
 \label{lemapprox}
\end{lem}

\appendix

\section{Transforming High Probability into Expected Distortion for Embeddings into Hamming Space}\label{secimplicitresultofCandK}

In this section, we prove Theorem \ref{thmprobtoexpected}, which is
implicitly present in Section 3.5 of \cite{UlamEmbedding}. It takes
any high-probability low-edit-distance-regime embedding from
edit-distance to Hamming space, and turns it into an embedding with
good expected distortion.

We will need the following routine lemma.
\begin{lem}
  Let $j \in \mathbb{N}$ and $k \in \mathbb{R}$ with $k \ge 2$. Then
  $$\frac{j}{4k} \le 1 - (1 - 1/k)^j \le \frac{j}{k}.$$
  \label{lembasicmath}
\end{lem}
\begin{proof}
Because $(1 - 1/k)^k$ is an increasing function in $k$ we must have
$(1 - 1/k)^k \ge (1 - 1/2)^2 = 1/4$. Therefore, for $i \le k + 1$, we have
that $\frac{1}{k} (1 - 1/k)^{i - 1} \ge \frac{1}{4k}$, and thus $(1 -
1/k) \cdot (1 - 1/k)^{i - 1} \le (1 - 1/k)^{i - 1} -
\frac{1}{4k}$. Applying this repeatedly, it follows that $(1 - 1/k)^j
\le 1 - \frac{j}{4k}$. On the other hand, when we compute $(1 -
1/k)^j$, each multiplication by $(1 - 1/k)$ subtracts at most $1/k$
from the previous term, giving that
  $$\frac{j}{4k} \le 1 - (1 - 1/k)^j \le \frac{j}{k}.$$
\end{proof}

We now prove Theorem \ref{thmprobtoexpected}. It is restated below.

\begin{thm}[Restatement of Theorem \ref{thmprobtoexpected}]
Let $S$ be a subset of $\Sigma^{\le n}$. Suppose we have a randomized
embedding $\phi:S \rightarrow \Sigma^m$ such that for all $x, y \in S$
satisfying $\ed(x, y) \le K$,
$$\ed(x, y) \le \frac{1}{T}\ham(\phi(x), \phi(y)) \le F(K) \ed(x, y)$$
with probability at least $1 - \frac{1}{K \cdot F(K)}$ for some
distortion function $F(K)$ with $F(K) \ge 2$, for some $T \ge 1$,
and for all $K \in \mathbb{N}$.

Then there exists $\alpha:S \rightarrow \{0, 1\}$ such that for all
$x, y \in S$ satisfying $\ed(x, y) \le K$,
$$\Omega\left(\ed(x, y)\right) \le K \cdot F(K) \Pr[\alpha(x) \neq
  \alpha(y)] \le O(\ed(x, y)F(K)).$$ In other words, $\alpha$ is an
embedding from edit distance (at most $K$) over $S$ to Hamming distance (scaled) with expected distortion at
most $O(F(K))$.
\end{thm}

\begin{proof}[Proof of Theorem \ref{thmprobtoexpected}]
  For $x \in S$, we define $\alpha(x)$ as follows. First define $w = \phi(x)$. Then, use (pairwise independent) hash
  functions $h_1, \ldots, h_m$ from $\Sigma$ to $\{0, 1\}$ to map $w$
  to a binary string $w' := h_1(w_1) \cdots h_m(w_m)$. Next, select a
  binary string $r \in \{0, 1\}^m$ such that each position in $r$
  takes value one with probability $\frac{1}{T \cdot K \cdot F(K)}$. Finally,
  define
  $$\alpha(x) := \langle r, w' \rangle \bmod 2 = \left(\sum_i r_i w'_i\right) \bmod 2.$$

  Now consider some $y \in \Sigma^n$ which is distinct from $x$ such
  that $\edit(x, y) \le K$. Then, following the same progression as
  for $x$, we define $u := \phi(y)$, $u'$ to be $u$ with its $i$-th
  letter hashed to $\{0, 1\}$ by $h_i$, and $\alpha(y) := \langle r,
  u'\rangle \bmod 2$.

  We wish to analyze $\Pr[\alpha(x) \neq \alpha(y)]$. Let $j$ be the
  number of positions in which $w$ and $u$ differ. There are two
  cases.
  \begin{itemize}
  \item Case 1 is that $\frac{j}{T}$ is not in the range $[\ed(x, y),
    F(K)\ed(x, y)]$. Because $w = \phi(x)$ and $u = \phi(y)$, this
    occurs with probability at most $\frac{1}{K \cdot
      F(K)}$. Therefore, this case affects $\Pr[\alpha(x) \neq
      \alpha(y)]$ by an additive factor of at most $\frac{1}{K \cdot
      F(K)}$.
  \item Case 2 is that $\ed(x, y) \le \frac{j}{T} \le F(K)\ed(x, y)$.

    If all of the ones in $r$ correspond with positions in which $w$
    and $u$ agree, then we are guaranteed that $\alpha(x) = \alpha(y)$. In
    order for this to occur, each of the $j$ positions in which $w$
    and $u$ disagree must be assigned a zero in $r$; this occurs with
    probability $\left(1 - \frac{1}{T \cdot K \cdot F(K)}\right)^j$.

    On the other hand, with probability $1 - \left(1 -
    \frac{1}{T \cdot K \cdot F(K)}\right)^j$, $r$ takes value one in at least one
    position in which $x$ and $y$ disagree. When this happens, we have $\Pr[\alpha(x) = \alpha(y)] =
    \frac{1}{2}$.

    Therefore, if we restrict ourselves to Case 2, then
    $$\Pr[\alpha(x) \neq \alpha(y)] = \left( 1 - \left(1 -
    \frac{1}{T \cdot K \cdot F(K)}\right)^j\right)\frac{1}{2}.$$

    By Lemma \ref{lembasicmath}, this implies
    $$\frac{1}{8} \cdot \frac{j}{T \cdot K \cdot F(K)} \le \Pr[\alpha(x) \neq \alpha(y)]
    \le \frac{1}{2} \cdot \frac{j}{T \cdot K \cdot F(K)}.$$

    Because $j$ is between $T \cdot \ed(x, y)$ and $T \cdot F(K) \cdot \ed(x, y)$ (inclusive),
    \begin{equation}
      \frac{1}{8} \cdot \frac{\ed(x, y)}{K \cdot F(K)} \le \Pr[\alpha(x) \neq
        \alpha(y)] \le \frac{1}{2} \cdot \frac{\ed(x, y)}{K}.
      \label{eqsand}
    \end{equation}

  \end{itemize}

  Since Case (2) occurs with probability at least $1 - \frac{1}{K
    \cdot F(K)} \ge \frac{1}{2}$, we can lower bound $\Pr[\alpha(x)
    \neq \alpha(y)]$ by $\frac{1}{2}\Pr[\alpha(x) \neq \alpha(y) \mid
    \text{ Case 2}]$. This give us
  $$\frac{1}{16} \cdot \frac{\ed(x,y)}{K \cdot F(K)} \le \Pr[\alpha(x) \neq \alpha(y)].$$

  Since Case (1) occurs with probability at most $\frac{1}{K \cdot F(K)}$, we
  can apply \eqref{eqsand} to bound $\Pr[\alpha(x) \neq \alpha(y)]$ above by
  $$\Pr[\alpha(x) \neq \alpha(y)] \le \frac{1}{K \cdot F(K)} + \frac{1}{2} \cdot \frac{\ed(x, y)}{K}.$$

  Multiplying our lower and upper bounds by $16K \cdot F(K)$, we get
  $$\edit(x, y) \le 16K \cdot F(K)\Pr[\alpha(x) \neq \alpha(y)] \le 16 + 8 \ed(x, y)F(K),$$
  completing the proof.
\end{proof}

\section{Proof of Lemma \ref{lemapprox}}\label{secproofoflemapprox}

In this section, we present a proof of Lemma \ref{lemapprox}.

Notice that for $x, y \in \Sigma^{\le n}$, there exists a sequence
$e_1, e_2, \ldots, e_{s}$ of edits from $x$ to $y$ such that
$s \le 2 \ed(x, y)$; such that each edit is an insertion or a
deletion; and such that for each $e_i$ none of the edits $e_1, e_2,
\ldots, e_{s}$ take place to the right of $e_i$. Define $x^0 = x$ and
$x^i$ to be $x$ after the first $i$ edits. In particular, $x^{s} =
y$. A letter in $x^i$ is \emph{untouched} if it was not involved in
any of the edits $e_1, \ldots, e_i$. For each untouched letter $l$ in
$x^i$, define $p(l)$ to be the position in $x$ where $l$ originates
(prior to the edits).

Given a sequence $E$ of edits to $\phi(x^i)$, a block in $\phi(x^i)$
is \emph{untouched} if $E$ does not edit it. Moreover, for an
untouched block $a$, we use $p_E(a)$ to denote the position which the
first letter of $a$ is mapped to by the edits in $E$.

A sequence of edits $E$ consisting only of insertions and deletions
from $\phi(x^i)$ to $\phi(x)$ is called \emph{nearly position
  preserving} if for each untouched block $a$ in $\phi(x^i)$ beginning
with an untouched letter $l$, we have that $|p(l) - p_E(a)| \le
2c$. In other words, the edits in $E$ map the block $a$ to a position
in $x$ which is within $2c$ of the position in $x$ to which the edits
$e_1, \ldots, e_i$ map $l$.

Moreover, for simplicity, we require from a nearly position preserving
sequence of edits $E$ that every block in $\phi(x^i)$ beginning with a
touched letter $l$ must be touched by $E$ (though this can be achieved
simply by deleting and re-inserting the block). This is simply so that
every untouched block $a$ is guaranteed to have a first letter $l$ for
which $p(l)$ is defined. Our constructions will satisfy this property
trivially.

Finally, define $d(x^i)$ to be the minimum size of a nearly position
preserving sequence of edits from $\phi(x^i)$ to $\phi(x)$.

\begin{lem}
  Let $m \in \mathbb{N}$ and consider some pair $x^t$ and $x^{t +
    1}$. Condition on the fact that for each $(m + 10) \cdot c$
  consecutive letters in $x^t$, there is some $h_k$ which is injective
  on the $8c$-letter substrings.

  Then, $\Pr[d(x^{t + 1}) - d(x^{t}) > m] \leq
  \left(\frac{1}{2}\right)^{\Omega(m)}$.
  \label{lemapproxalgphiaftersingleinsertion}
\end{lem}
\begin{proof}
  We will consider only the case where $x^{t + 1}$ can be derived from
  $x^t$ by an insertion. The case where $x^{t + 1}$ is obtained from
  $x^t$ by a deletion follows by the exact same reasoning (except with
  some indices adjusted by two to account for the fact that they take
  place after a deletion rather than an insertion).

  We will closely follow the proof of Lemma
  \ref{lemphiaftersingleinsertion} with a few small modifications. Let
  $E$ be an minimum nearly position preserving sequence of insertions
  and deletions transforming $x^t$ to $x$. In particular, $|E| =
  d(x^t)$. Our goal is to show that with probability $1 -
  \frac{1}{2}^{\Omega(m)}$ we can construct a nearly position
  preserving sequence of $|E| + O(m)$ edits from $x^{t + 1}$ to $x$.

  Let $i$ be the index such that $x^{t + 1}$ is generated from $x^t$
  by adding a character at index $i$. That is, $x^t_1 \cdots x^t_{i-1}
  = x^{t + 1}_1 \cdots x^{t +1}_{i-1}$ and $x^t_i x^t_{i+1} \cdots =
  x^{t + 1}_{i+1} x^{t + 1}_{i+2} \cdots$.

  Call a block in $\phi(x^t)$ (resp. $\phi(x^{t + 1})$) a
  \emph{beginning block} if it begins before or at $x^t_i$
  (resp. $x^{t + 1}_{i + 1}$) and an \emph{ending block} if it begins
  after $x^t_i$ (resp. $x^{t + 1}_{i + 1}$). Moreover, we call a block
  in $\phi(x)$ a \emph{beginning block} if it begins before or at
  $x_{p(x^t_i)}$, and an \emph{ending block} if it begins after
    $x_{p(x^t_i)}$.

  Because $E$ is nearly position preserving, any untouched ending
  block $u$ must be mapped by $E_1$ to begin in a position at least
  $p(x^t_i) - 2c$. Therefore, there is a subset $E'$ of the edits of
  $E$ that maps the beginning blocks in $\phi(x^t)$ to a prefix of $x$
  containing at least $x_1 \cdots x_{p(x^t_i) - 2c - 1}$. Moreover,
  because $E$ is nearly position preserving, any untouched beginning
  block $u$ in $\phi(x^t)$ must be mapped by $E$ to begin in a
  position no larger than $p_t(x^t_i) + 2c$. Therefore, there is a
  subset $E''$ of $E'$ mapping the beginning blocks in $\phi(x^t)$ to
  a prefix of $x$ containing at least $x_1 \cdots x_{p(x^t_i) - 2c}$
  but containing no blocks in $x$ which begin after $x_{p(x^t_i) +
    2c}$.\footnote{In particular, Since $E'$ does not map any unedited
    blocks to begin in a position greater than $x_{p(x^t_i) - 2c}$,
    the only way it can result in such blocks is insert them. $E''$
    can therefore be obtained from $E'$ by ignoring such insertions.}

  From the logic in the proof of Lemma
  \ref{lemphiaftersingleinsertion}, the blocks that contain only
  characters before $x^{t}_{i-16c-c/2 - 1}$ in $\phi(x^t)$ are identical to
  those which appear before $x^{t + 1}_{i-16c-c/2 - 1}$ in $\phi(x^{t +
    1})$. By Corollary \ref{corboundlength}, it follows that $O(1)$
  deletions and insertions can be used to transform the beginning
  blocks of $\phi(x^{t + 1})$ to the beginning blocks of
  $\phi(x^t)$. Then using the edits from $E''$, we can transform these
  blocks into a prefix of $x$ containing at least $x_1 \cdots
  x_{p(x^t_i) - 2c}$ but containing no blocks in $x$ which begin after
  $x_{p(x^t_i) + 2c}$. Finally, (because of Corollary
  \ref{corboundlength}) with $O(1)$ additional edits, we can obtain the
  beginning blocks of $x$.

  So far we have shown how to perform a series of $|E| + O(1)$
  insertions and deletions to transform the beginning blocks of
  $x^{t + 1}$ to the beginning blocks of $x^t$ in a way such that any
  untouched blocks have their first letter $l$ mapped to a position
  within $2c$ of $p_{t + 1}(l)$.

  Next we will focus on the ending blocks of $x^{t + 1}$. We remark
  that we will construct a sequence of edits between these blocks and
  the ending blocks of $x$ without any further use of $E$. Just as in
  the proof of Lemma \ref{lemphiaftersingleinsertion}, there are two
  cases to consider.

  \textbf{Case 1: }$x^t_i$ does not belong to a maximally periodic
  substring that extends until at least $x^t_{i+8c-1}$. Following the
  same logic as in Lemma \ref{lemphiaftersingleinsertion}, with
  probability at least $\left(\frac{1}{2}\right)^{m - 2}$, there will
  be some $j$ between $i$ and $i + (m + 18 + 1/2)c$ such that the
  blocks in $\phi(x^{t + 1})$ containing letters only from $x^{t +
    1}_{j + 1}x^{t + 1}_{j + 2} \cdots$ are identical to the blocks of
  $\phi(x^t)$ containing letters only from $x^{t}_jx^t_{j + 1}
  \cdots$. Moreover, because each of the edits from $x^0$ to $x^{t +
    1}$ take place before or at $x^{t + 1}_i$, the same logic can be
  used to conclude that the blocks in $\phi(x^{t + 1})$ containing
  letters only from $x^{t + 1}_{j + 1}x^{t + 1}_{j + 2} \cdots$ are
  identical to the blocks of $\phi(x)$ containing letters only from
  $x_{p(x^{t + 1}_{j + 1})}x_{p(x^{t + 1}_{j + 1}) + 1} \cdots$.

  Notice that $p(x^{t + 1}_{j + 1}) = p(x^{t}_{i}) + (j - i)$ because
  all of the edits from $x^0$ to $x^{t + 1}$ occur prior to $x^{t +
    1}_{i + 1}$, which is the same letter as $x^t_i$. It follows that
  $p_{t + 1}(x^{t + 1}_j) \le p_t(x^t_i) + O(mc)$. Therefore, by Lemma
  \ref{lemphiaftersingleinsertion}, only $O(m)$ edits are needed to
  $\phi(x^{t + 1})$ to remove the ending blocks which contain letters
  prior to $x^{t + 1}_{j + 1}$ and replace them with the ending blocks
  in $x$ which contain letters prior to $x_{p(x^{t + 1}_{j +
      1})}$. Since the remaining ending blocks are the same between
  $x^{t + 1}$ and $x$, this transforms the ending blocks of $\phi(x^{t
    + 1})$ into those of $x$. Combining this with the edits which
  transform the beginning blocks of $x^{t + 1}$ to those $x$, we have
  a nearly position preserving series of at most $|E| + O(m)$ edits
  transforming $x^{t + 1}$ to $x$.

  \textbf{Case 2: }$x^t_i$ belongs to a maximally periodic substring
  that extends until at least $x^{t}_{i+8c-1}$. Denote by $x^t_{i^*}$
  the last character of $x^t$ to be in a maximally periodic substring
  which contains $x^t_i$. Then by the logic from the proof of Lemma
  \ref{lemphiaftersingleinsertion}, the blocks beginning at $x^t_{i^*
    + 1}$ will be identical to those beginning at $x^{t + 1}_{i^* +
    2}$. Moreover, because all of the edits from $x^0$ to $x^{t}$
  occur prior to $x^t_i$, the same logic can be used to conclude that
  the blocks beginning at $x_{p(x^{t + 1}_{i^* + 2})}$ are identical
  to those beginning at $x^{t + 1}_{i^* + 2}$.

  So far we have shown that the beginning blocks of $\phi(x^{t + 1})$
  can be transformed to the beginning blocks of $\phi(x)$ through a
  nearly position preserving sequence of at most $|E| + O(1)$ edits,
  and that the ending blocks of $\phi(x^{t + 1})$ beginning with $x^{t
    + 1}_{i^* + 2}$ are identical to those of $\phi(x)$ beginning with
  $x_{p(x^{t + 1}_{i^* + 2})}$. Therefore, it suffices to show that
  the ending blocks of $\phi(x^{t + 1})$ preceding $x^{t + 1}_{i^* +
    2}$ can be transformed to those in $\phi(x)$ preceding $x_{p(x^{t
      + 1}_{i^* + 2})}$ through a nearly position preserving sequence
  of $O(1)$ edits. Denote by $j_1$ the index of the first marker
  placed in $x_{p(x^{t + 1}_{i + 1}) + 2c} \cdots x_{p(x^{t +
      1}_{i^*}) - 2c}$, and by $j_1'$ the index of the first marker
  placed in $x^{t + 1}_{i+2c+1} \cdots x^{t + 1}_{i^*- 2c +
    1}$. Additionally, denote by $j_2$ the index of the last marker
  placed before or at $x_{p(x^{t + 1}_{i^*}) - 2c}$ and by $j_2'$ the
  index of the last marker placed before or at $x^{t + 1}_{i^* - 2c
    +1}$. By the same logic as from Lemma
  \ref{lemphiaftersingleinsertion}, each of the blocks in $\phi(x)$
  between $x_{j_1}$ and $x_{j_2}$ and each of the blocks in $\phi(x^{t
    + 1})$ between $x^{t + 1}_{j_1'}$ and $x^{t + 1}_{j_2'}$ are
  identical. Moreover, notice because all of the edits $e_1, \ldots,
  e_{t}$ take place before $x^{t + 1}_{i + 1}$, we have that $p(x^{t +
    1}_{i + 1 + r}) = p(x^{t + 1}_{i + 1}) + r$ for all $r \ge 0$. It
  follows that for all but the first and last $O(c)$ letters of $x^{t
    + 1}_{j_1'} \cdots x^{t + 1}_{j_2'}$, each letter $l$ has the
  property that $x_{p(l)}$ appears in $x_{j_1} \cdots x_{j_2}$. Hence,
  by Corollary \ref{corboundlength}, we can perform $O(1)$ edits to
  transform the blocks between $x^{t + 1}_{j_1'}$ and $x^{t +
    1}_{j_2'}$ to the blocks between $x_{j_1}$ and $x_{j_2}$ in a way
  so that any untouched block beginning with some letter $l$ is mapped
  to the block in $x$ containing $x_{p_{t + 1}(l)}$; it follows these
  edits are nearly position preserving.

  By Corollary \ref{corboundlength} all but $O(1)$ of the ending
  blocks in $x^{t + 1}$ appear either between $x^{t + 1}_{j_1'}$ and
  $x^{t + 1}_{j_2'}$ or after $x^{t + 1}_{i^* + 2}$, and all but
  $O(1)$ of the ending blocks in $x$ appear either between $x_{j_1}$
  and $x_{j_2}$ or after $x_{p(x^{t + 1}_{i^* + 2})}$. Therefore, with $O(1)$ additional
  edits, we can eliminate such blocks in $x^{t + 1}$ and replace them
  with the appropriate blocks in $x$. This completes the nearly
  position preserving progression of at most $|E| + O(1)$ edits from
  $x^{t + 1}$ to $x$.
\end{proof}

\begin{lem}
  Suppose $c \le n^{0.4}$. Let $x$ and $y$ be strings. Then for $m \in
  \mathbb{N}$,
  $$\Pr[d(x^s) > m\ed(x, y)] \leq \left(\frac{1}{2}\right)^{\Omega(m)}
  + \frac{1}{\poly(n)},$$ where $\poly(n)$ is an arbitrarily large
  polynomial in $n$ of our choice controlled by the constant $b$.
  \label{lemalmostlemapprox}
\end{lem}
\begin{proof}
  It suffices to show that
    $$\Pr\left[\sum_{i = 0}^{s - 1} \left(d(x^{i + 1}) - d(x^{i})\right) \le m\ed(x, y)\right]
  \leq \left(\frac{1}{2}\right)^{\Omega(m)} + \frac{1}{\poly(n)}.$$

  Lemma \ref{lemapproxalgphiaftersingleinsertion} tells us that for
  each $i \le s$, if we condition on the fact that each $(m + 9) \cdot
  c$ consecutive letters in $x^i$ has some $h_k$ which is injective on
  its $8c$-letter substrings, then $\Pr[d(x^{i + 1}) - d(x^{i}) > m]
  \leq \left(\frac{1}{2}\right)^{\Omega(m)}$. Notice that, in
  comparison, Lemma \ref{lemphiaftersingleinsertion} tells us that for
  under the same conditions, $\Pr[\ed(x^i, x^{i + 1}) > m] \leq
  \left(\frac{1}{2}\right)^{\Omega(m)}$. Thus we can use the same
  proof as was used for Lemma \ref{lemboundexpansion}, except with
  Lemma \ref{lemapproxalgphiaftersingleinsertion} replacing Lemma
  \ref{lemphiaftersingleinsertion}. In particular, whereas the proof
  previously could be used to obtain the bound
  $$\Pr \left[\sum_{i = 0}^{s - 1}\ed(x^i, x^{i + 1}) \ge m \cdot
    \ed(x, y) \right] \le \frac{1}{2^{\Omega(m)}} + \frac{1}{\poly(n)},$$
  the same proof can now be used to provide the bound
    $$\Pr\left[\sum_{i = 0}^{s - 1} \left(d(x^{i + 1}) - d(x^{i})\right)
    \ge m\ed(x, y)\right] \le\left(\frac{1}{2}\right)^{\Omega(m)} +
  \frac{1}{\poly(n)},$$ as desired.
\end{proof}

We are now in a position to prove Lemma \ref{lemapprox}.
\begin{lem}[Restatement of Lemma \ref{lemapprox}]
 Consider $x, y \in \Sigma^{\le n}$. Pick $m \in \mathbb{N}$. Then
 with probability $1- \left(\frac{1}{2}\right)^{\Omega(m)} -
 \frac{1}{\poly(n)}$ (for a polynomial of our choice), there is a
 sequence of at most $m \ed(x, y)$ edits from $\phi(x)$ (with
 parameter $c$) to $\phi(y)$ with the following property. If a block
 $u$ in $\phi(x)$ is not modified by the edits and is mapped to a
 block $v$ in $\phi(y)$, then the start position of block $u$ in $x$
 is within $O(\ed(x, y) + c)$ of the start position of block $v$ in
 $y$.
\end{lem}
\begin{proof}[Proof of Lemma \ref{lemapprox}]
  Lemma \ref{lemalmostlemapprox} establishes that with probability $1-
  \left(\frac{1}{2}\right)^{\Omega(m)} - \frac{1}{\poly(n)}$ there is
  a nearly position preserving sequence of edits from $\phi(y)$ to
  $\phi(x)$. Such a sequence of edits must map each untouched block
  $u$ in $\phi(y)$ starting with some letter $y_i$ to a position in
  $x$ within $O(c)$ of $p(y_i)$. Since $p(y_i)$ can differ from $i$ by
  at most $\ed(x, y)$, it follows that if a block $u$ in $\phi(x)$ is
  not modified by the edits and is mapped to a block $v$ in $\phi(y)$,
  then the start position of block $u$ in $x$ is within $O(\ed(x, y) +
  c)$ of the start position of block $v$ in $y$.
\end{proof}

\section{Omitted Proofs}\label{proofsappendix}

\printproofs

\bibliographystyle{plain}
\bibliography{edit_distance_ref}

\begin{thebibliography}{10}

\bibitem{blast}
Stephen~F. Altschul, Thomas~L. Madden, Alejandro~A. Sch{\"a}ffer, Jinghui
  Zhang, Zheng Zhang, Webb Miller, and David~J. Lipman.
\newblock Gapped blast and psi-blast: a new generation of protein database
  search programs.
\newblock {\em Nucleic acids research}, 25(17):3389--3402, 1997.

\bibitem{LowerBoundsecond}
Alexandr Andoni and Robert Krauthgamer.
\newblock The computational hardness of estimating edit distance.
\newblock {\em {SIAM} J. Comput.}, 39(6):2398--2429, 2010.

\bibitem{ApproxPolyLog}
Alexandr Andoni, Robert Krauthgamer, and Krzysztof Onak.
\newblock Polylogarithmic approximation for edit distance and the asymmetric
  query complexity.
\newblock In {\em Proceedings of the 51st Annual Symposium on Foundations of
  Computer Science (FOCS)}, pages 377--386. IEEE, 2010.

\bibitem{ApproxSubPolyDistortion}
Alexandr Andoni and Krzysztof Onak.
\newblock Approximating edit distance in near-linear time.
\newblock {\em {SIAM} J. Comput.}, 41(6):1635--1648, 2012.

\bibitem{backurs2015edit}
Arturs Backurs and Piotr Indyk.
\newblock Edit distance cannot be computed in strongly subquadratic time
  (unless seth is false).
\newblock In {\em Proceedings of the 47th Annual Symposium on Theory of
  Computing (STOC)}, pages 51--58. ACM, 2015.

\bibitem{bar2004approximating}
Ziv Bar-Yossef, T.S. Jayram, Robert Krauthgamer, and Ravi Kumar.
\newblock Approximating edit distance efficiently.
\newblock In {\em Proceedings of 45th Annual Symposium on Foundations of
  Computer Science (FOCS)}, pages 550--559. IEEE, 2004.

\bibitem{DefnExpectedDistortion}
Yair Bartal.
\newblock On approximating arbitrary metrices by tree metrics.
\newblock In {\em Proceedings of the 30th Annual Symposium on the Theory of
  Computing (STOC)}, pages 161--168, 1998.

\bibitem{DimensionReduction}
Tugkan Batu, Funda Erg{\"{u}}n, and S{\"{u}}leyman~Cenk Sahinalp.
\newblock Oblivious string embeddings and edit distance approximations.
\newblock In {\em Proceedings of the 17th Annual Symposium on Discrete
  Algorithms (SODA)}, pages 792--801, 2006.

\bibitem{belazzougui2016edit}
Djamal Belazzougui and Qin Zhang.
\newblock Edit distance: Sketching, streaming, and document exchange.
\newblock In {\em Proceedings of the 57th Annual Symposium on Foundations of
  Computer Science (FOCS)}, pages 51--60. IEEE, 2016.

\bibitem{chakraborty2016streaming}
Diptarka Chakraborty, Elazar Goldenberg, and Michal Kouck{\`y}.
\newblock Streaming algorithms for embedding and computing edit distance in the
  low distance regime.
\newblock In {\em Proceedings of the 48th Annual Symposium on Theory of
  Computing (STOC)}, pages 712--725. ACM, 2016.

\bibitem{heuristicoriginal}
Kun-Mao Chao, William~R. Pearson, and Webb Miller.
\newblock Aligning two sequences within a specified diagonal band.
\newblock {\em Bioinformatics}, 8(5):481--487, 1992.

\bibitem{UlamEmbedding}
Moses Charikar and Robert Krauthgamer.
\newblock Embedding the ulam metric into \emph{l}\({}_{\mbox{1}}\).
\newblock {\em Theory of Computing}, 2(11):207--224, 2006.

\bibitem{CLRS}
Thomas~H. Cormen, Charles~E. Leiserson, Ronald~L. Rivest, and Clifford Stein.
\newblock {\em Introduction to algorithms}.
\newblock The MIT Press, third edition, 2009.

\bibitem{farach1997optimal}
Martin Farach.
\newblock Optimal suffix tree construction with large alphabets.
\newblock In {\em Proceedings of the 38th Annual Symposium on Foundations of
  Computer Science (FOCS)}, pages 137--143. IEEE, 1997.

\bibitem{RangeMinimumQuery}
Johannes Fischer and Volker Heun.
\newblock Theoretical and practical improvements on the rmq-problem, with
  applications to {LCA} and {LCE}.
\newblock In {\em Proceedings of the 17th Annual Symposium on Combinatorial
  Pattern Matching (CPM)}, pages 36--48, 2006.

\bibitem{gusfield1997algorithms}
Dan Gusfield.
\newblock {\em Algorithms on strings, trees and sequences: computer science and
  computational biology}.
\newblock Cambridge University Press, 1997.

\bibitem{SlidingWindow}
Richard Harter.
\newblock The minimum on a sliding window algorithm.
\newblock \url{http://www.richardhartersworld.com/cri/2001/slidingmin.html}.
\newblock Accessed: 2017-07-10.

\bibitem{indyk2001algorithmic}
Piotr Indyk.
\newblock Algorithmic aspects of geometric embeddings (tutorial).
\newblock In {\em Proceedings of the 42nd Annual Symposium on Foundations of
  Computer Science (FOCS)}, pages 10--33. IEEE, 2001.

\bibitem{HashReference}
Piotr Indyk.
\newblock A small approximately min-wise independent family of hash functions.
\newblock {\em J. Algorithms}, 38(1):84--90, 2001.

\bibitem{indyk20048}
Piotr Indyk and Jiri Matou{\v{s}}ek.
\newblock Low-distortion embeddings of finite metric spaces.
\newblock {\em Handbook of Discrete and Computational Geometry}, page 177,
  2004.

\bibitem{LowerBoundfirst}
Robert Krauthgamer and Yuval Rabani.
\newblock Improved lower bounds for embeddings into \emph{L}\({}_{\mbox{1}}\).
\newblock In {\em Proceedings of the 17th Annual Symposium on Discrete
  Algorithms (SODA)}, pages 1010--1017, 2006.

\bibitem{landau1998incremental}
Gad~M. Landau, Eugene~W. Myers, and Jeanette~P. Schmidt.
\newblock Incremental string comparison.
\newblock {\em SIAM Journal on Computing}, 27(2):557--582, 1998.

\bibitem{navarro2001guided}
Gonzalo Navarro.
\newblock A guided tour to approximate string matching.
\newblock {\em ACM computing surveys (CSUR)}, 33(1):31--88, 2001.

\bibitem{needleman1970general}
Saul~B. Needleman and Christian~D. Wunsch.
\newblock A general method applicable to the search for similarities in the
  amino acid sequence of two proteins.
\newblock {\em Journal of molecular biology}, 48(3):443--453, 1970.

\bibitem{Embedding}
Rafail Ostrovsky and Yuval Rabani.
\newblock Low distortion embeddings for edit distance.
\newblock {\em J. {ACM}}, 54(5):23, 2007.

\bibitem{amazinghash}
Anna Pagh and Rasmus Pagh.
\newblock Uniform hashing in constant time and optimal space.
\newblock {\em SIAM Journal on Computing}, 38(1):85--96, 2008.

\bibitem{vintsyuk1968speech}
Taras~K. Vintsyuk.
\newblock Speech discrimination by dynamic programming.
\newblock {\em Cybernetics}, 4(1):52--57, 1968.

\bibitem{WagnerF74}
Robert~A. Wagner and Michael~J. Fischer.
\newblock The string-to-string correction problem.
\newblock {\em J. {ACM}}, 21(1):168--173, 1974.

\end{thebibliography}

\end{document}